\documentclass[a4paper,11pt]{article}
\pdfoutput=1 
\newif\ifclean
\cleantrue
\cleanfalse

\usepackage{fancyhdr}
\usepackage{etoolbox}
\usepackage{eurosym}
\usepackage{amsmath,amsfonts,amssymb,amsthm}
\usepackage{stmaryrd}
\usepackage{xspace}
\usepackage[dvipsnames]{xcolor}
\usepackage{msc}
\usepackage{calc}
\usepackage{multirow}
\usepackage{epsfig}
\usepackage{pstricks,pst-node}
\usepackage{graphicx}
\usepackage{boldline}
\usepackage{pifont}
\usepackage[a4paper, margin=2.5cm]{geometry}
\usepackage{todonotes}
\usepackage{caption}
\usepackage{subcaption}
\usepackage{proof}
\usepackage{listings}
\usepackage{pdfcomment}
\usepackage{cleveref}
\usepackage{physics}
\usepackage{qcircuit}
\usepackage{authblk}  

\newcommand{\ie}{\textit{i.e.,} }
\newcommand{\eg}{\textit{e.g.,} }

\newcommand{\senc}{\mathsf{senc}}
\newcommand{\mac}{\mathsf{mac}}

\newcommand{\exenc}{\mathsf{enc}}
\newcommand{\exdec}{\mathsf{dec}}

\newcommand{\pair}[2]{\langle #1;#2\rangle}

\newcommand{\bit}{\mathsf{bit}}
\newcommand{\qubit}{\mathsf{qubit}}

\definecolor{gris}{gray}{0.85}

\newtheorem{example}{Example}
\newtheorem{definition}{Definition}

\makeatletter
\newcommand{\xRightarrow}[2][]{\ext@arrow 0359\Rightarrowfill@{#1}{#2}}
\makeatother
\makeatletter
\newcommand{\rightdoublearrow}{%
  \rightarrow\mkern-10mu\protect\joinrel\rightarrow}

\newcommand{\rightdoublearrowfill@}
  {\arrowfill@\relbar\relbar\rightdoublearrow}
\newcommand{\mapstodoublearrowfill@}
  {\arrowfill@{\mapstochar\relbar}\relbar\rightdoublearrow}
\newcommand{\xrightdoublearrow}[2][]
  {\ext@arrow 3{15}59\rightdoublearrowfill@{#1}{#2}}
\newcommand{\xmapstodoublearrow}[2][]
  {\ext@arrow 3{15}59\mapstodoublearrowfill@{#1}{#2}}
\makeatother

\newcommand{\akiss}{\textsc{Akiss}\xspace}
\newcommand{\tamarin}{\textsc{Tamarin}\xspace}

\newcommand{\proverif}{\textsc{ProVerif}\xspace}
\newcommand{\deepsec}{\textsc{DeepSec}\xspace}

\newcommand{\etc}{\textit{etc.}\xspace}

\makeatletter
\def\rightarrowfillstar@{\arrowfill@\relbar\relbar{\rightarrow\smash{^*}}}
\newcommand{\xrightarrowstar}[2][]{\ext@arrow
  0{13}{15}8\rightarrowfillstar@{#1}{#2}}
\newcommand{\lrstep}{\@ifstar{\xrightarrowstar}{\xrightarrow}}
\makeatother

\newcommand{\N}{\mathcal{N}}

\newcommand{\theo}{\approx_E}

\newtheorem{postulate}{Postulate}
\newtheorem{remark}{Remark}
\newtheorem{theorem}{Theorem}
\newtheorem{attack}{Attack}
\newtheorem{abstraction}{Abstraction}
\crefname{abstraction}{Abstraction}{Abstractions}
\crefname{theorem}{Theorem}{Theorems}
\crefname{remark}{Remark}{Remarks}
\crefname{attack}{Attack}{Attacks}
\crefname{postulate}{Postulate}{Postulates}

\renewcommand{\dag}{^\dagger}
\newcommand{\deff}{\triangleq}

\definecolor{mygreen}{rgb}{0,0.6,0}
\definecolor{mygray}{rgb}{0.4,0.4,0.4}
\definecolor{mymauve}{rgb}{0.6,0,0.9}
\lstdefinelanguage{Tamarin}
{
  sensitive=true,
  keywordstyle=[1]{\color{blue}},
  keywordstyle=[0]{\tt \color{mygreen}},
  keywordstyle=[2]{\color{mymauve}},
  keywordstyle=[3]{\color{red}},
  commentstyle = {\color{mygray}},
  stringstyle = {\color{mygreen}},
  otherkeywords={--[,]->,-->,--,>,<,[,],=,!,:,$},
  morekeywords=[0]{~,$},
  morekeywords=[1]{let,in,Fr,In,Out,!},
  morekeywords=[2]{--[,--,>,<,-,]->,-->,[,],signing,=},
  morekeywords=[3]{restriction,lemma,rule,bultins,functions,equations,:},
  morecomment=[l]{//}, 
  morecomment=[s]{/*}{*/}, 
  morestring=[b]', 
}
\newcommand{\factStyle}[1]{\textsf{#1}}

\newcommand{\rwrleft}{\mathrel{-\!\!\!\!-\!\!\![}}
\newcommand{\rwrright}{\mathrel{]\!\!\!\rightarrow}}
\newcommand{\rwr}[1][]{{\mathrel{\rwrleft #1 \rwrright}}}

\title{Symbolic Abstractions for Quantum Protocol Verification}
\author{Lucca Hirschi\footnote{This work was partially done while the author was at ETH Zurich, Switzerland.}}
\affil{Inria \& LORIA, Nancy, France\\
          \href{mailto:lucca.hirschi@inria.fr}{lucca.hirschi@inria.fr}}

\begin{document}
\date{}
\maketitle

\paragraph*{Abstract.}
\looseness=-1
Quantum protocols such as the BB84 Quantum Key Distribution protocol exchange qubits
to achieve information-theoretic security guarantees. Many variants thereof were proposed,
some of them being already deployed.
Existing security proofs in that field are mostly tedious, error-prone pen-and-paper proofs of the core
protocol only that rarely account for other crucial components such as authentication.
This calls for formal and automated verification techniques
that exhaustively explore all possible intruder behaviors and that scale well.

The {\em symbolic approach} offers rigorous, mathematical frameworks and automated tools
to analyze security protocols. Based on well-designed abstractions, it has allowed
for large-scale formal analyses of 
real-life protocols such as TLS 1.3 and mobile telephony protocols to be conducted fully automatically.

Hence a natural question is:
{\em Can we use this successful line of work to analyze quantum protocols?}
This paper proposes a first positive answer and motivates further research on this unexplored path.
\smallskip{}

\tableofcontents

\section{Introduction}
\subsection{Background}

\paragraph{Quantum Physics.}
Quantum physics usually distinguishes between {\em quantum computation} and {\em quantum information}.
The former focuses on new algorithms and computation capabilities that require quantum computers
while the latter embraces new information exchange mechanisms enabled by quantum physics.
When focusing on cryptography, the same dichotomy exists: {\em post-quantum cryptography} refers to new cryptography
schemes that are expected to be resistant against quantum computers while {\em quantum protocols} are new security protocols
that rely on quantum information mechanisms to achieve security properties, which are often information-theoretic rather than conditional.
The present work focuses on quantum protocols.

\paragraph{Quantum Protocols.}
\looseness=-1
Quantum protocols such as the BB84 Quantum Key Distribution (QKD) protocol~\cite{bennett1984quantum} or its bit commitment
variant~\cite{brassard1990quantum} exchange qubits to achieve extremely strong security guarantees.
Namely, quantum physics ensures that observing the exchanged qubits disturbs the observed data. This can be exploited by the parties involved to effectively
detect the presence of an eavesdropper, which in turn can be leveraged to achieve information-theoretic security guarantees.

\looseness=-1
Many different quantum protocols, some already deployed and commercialized\footnote{See for example \url{www.idquantique.com}.},
essentially rely on the same mechanisms and can be seen as variants of the BB84 QKD protocol (\eg see ~\cite{bruss2007quantum}).
For instance, the Quantum Bit-Commitment (QBC) variant of BB84~\cite{brassard1990quantum} was proposed in 1990 and was
believed to be secure (see~\cite{brassard1993quantum}) until
it was shown to be flawed 7 years later~\cite{lo1997quantum} due to an EPR-attack, \ie based on entangled qubits.
The QKD BB84 protocol, however, is still believed to be secure and various pen-and-paper proofs of
unconditional security exist~\cite{scarani2009security}.
However, such proofs often hold for the core protocol only and rarely account for other crucial components such as authentication or authorization, \cite{portmann2014cryptographic} being a notable exception.
Furthermore, we stress that in such complex settings, where intruders have different capabilities that can be combined in
infinitely many  ways, manual proofs are highly error-prone.
Finally, numerous variants of such protocols have been given or will be proposed in the future, making tedious and error-prone manual proofs impractical.

The proliferation of protocols and the complexity of their proofs call for formal and automated verification techniques
that exhaustively explore all possible intruder behaviors.
%
%
Automation usually comes at the cost of approximations and less precise security guarantees.
In this paper, we explore this trade-off.



\subsection{State-of-the-Art}

\paragraph{Automated Formal Verification.}
{\em Formal methods} offer rigorous, mathematical frameworks to analyze security
protocols.
Two main approaches have emerged to provide mathematical foundations
for this analysis,
starting with the seminal works of~\cite{dolev1983security,GOLDWASSER1984270}:
the {\em computational approach} and
the {\em symbolic approach}.

The computational approach is based on 
the {\em standard model}:
messages are modeled as bit strings,
and agents and the attacker as probabilistic polynomial time Turing machines.
Security goals are then defined using games played by the attacker
and proofs are usually done via reductions (or hops) between
successive games until reaching games expressing computational assumptions on cryptographic primitives.
It is generally acknowledged that security proofs in this model offer strong security guarantees.
However, a serious downside of this approach is that even for small protocols, the proofs are usually difficult,
tedious, and error prone. Moreover,
due to the high complexity of this model, automating such proofs is a difficult problem
that is still in its infancy. 
More generally, computer-aided verification allows for only a low level of automation, even though considerable efforts have been
put in developing verifiers
such as 
CertiCrypt~\cite{barthe2009formal},
EasyCrypt~\cite{barthe2011computer},
FCF~\cite{beringer2015verified,petcher2015foundational},
and F$^*$\cite{barthe2014probabilistic}.

In contrast to the computational approach, the {\em symbolic approach} is used when one is interested in analyzing in a reasonable time
more complex protocols, rather than simple primitives or core protocols.
This model is more abstract and scales better. In particular, it makes strong assumptions on cryptographic primitives
(\ie the perfect cryptography assumption) but fully models algebraic properties of these primitives as well as the protocol agents'
interactions.
Modeling security protocols using the symbolic approach allows one to benefit from
machine support using established techniques, such as
model-checking, resolution, and rewriting techniques.
From the different lines of work in this area there have emerged verification
tools (\eg \tamarin~\cite{TAMARIN-CAV}, \proverif~\cite{PROVERIF}, \deepsec~\cite{DEEPSEC}
) 
and large-scale formal analyses of 
real-life protocols (\eg\linebreak[4]  
TLS 1.3~\cite{cremers2016automated,TLS-PV-CP,TLS-Cas-CCS17}, 
mobile telephony protocols~\cite{5GAKA},
instant messaging protocols~\cite{Signal-PV-CP}, 
and entity authentication protocols~\cite{basin2013provably}).

\looseness=-1
Hence a natural question is:
{\em Can we use this successful line of work to analyze quantum protocols?}
This paper proposes a first positive answer and motivates further research on this unexplored path.

\paragraph{Automated Verification of Quantum Protocols.}
%
%
%
In the standard model, \cite{unruh2019quantum} proposes an extension of the pRHL logic (pRHL is used by EasyCrypt~\cite{barthe2011computer}
and CertiCrypt~\cite{barthe2009formal}) to handle quantum protocols and post-quantum cryptography  schemes. They provide
a tool that produces machine-checked security
proofs. While such techniques aim at establishing extremely strong security guarantees, they inherit the complexity of the computational model
and its low level of automation.
In this paper, we focus on automated verification techniques, even if this means providing weaker guarantees.

Other research explores the use of model checkers for probabilistic distributed programs.
\cite{nagarajan2005automated}~models and analyzes the BB84 QKD protocol using the probabilistic model
checker\linebreak[4] PRISM~\cite{kwiatkowska2001automated}.
While such analyses can quantify the probability for the intruder to be detected or to learn the key,
this approach does not consider a fully adversarial environment but rather considers
a fixed, limited intruder behavior; namely a receive-resend behavior.
%
Abstracting away probabilities, prior works have also used non-probabilistic model-checking tools for distributed programs.
\cite{nagarajan2002formal} verifies, in the presence of a fixed intercept-resend attacker,
that the BB84 QKD protocol is trace equivalent to its specification.
In contrast, \cite{belardinelli2012automated} is interested in verifying safety properties in a non-adversarial environment,
expressed in epistemic logic, using model-checking for multi-agents systems. 
Similarly, \cite{ardeshir2014verification} proposes a verification technique for checking equivalence between
quantum protocols.
None of these works considers a fully adversarial environment and they all fall short of capturing other potential cryptographic components upon which the core quantum protocol is based (\eg authentication).

Symbolic models have been successfully used in the past to model classical security protocols in a fully adversarial setting.
They cannot however be used off-the-shelf to analyze quantum protocols.
The main features thereof that are not handled by classical techniques are:
the intruder's quantum capabilities that are limited by quantum physics laws (\eg no-cloning, measurement) and the
protocol logic conditioned by probabilities, and security parameters.
%

\subsection{Our Contributions.}
\looseness=-1
We formally define a novel extension of the classical Dolev-Yao intruder
accounting for some quantum physics capabilities and extending
the intruder's control to the quantum channels.  Our attacker can
read qubits, produce new qubits, and produce entangled
qubits in order to perform EPR-attacks.
However, his capabilities are restricted
by standard Quantum Physics principles, for example the no-cloning theorem and the
Heisenberg uncertainty principle. Hence, our framework accounts for
a fully adversarial environment with regard to a rich class of
intruder capabilities, as opposed to a fixed, trivial intruder
strategy.  However, due to known limitations of the symbolic model,
our extension does not capture probabilistic attacks in their
full generality.  Still, our extension does capture {\em a class} of
probabilistic attacks by allowing the attacker to guess a fix
amount of the bits that are randomly chosen by honest parties.

\looseness=-1
We show how our quantum Dolev-Yao attacker can be embedded in some classical verification tools for cryptographic protocols
such as \tamarin, \proverif, and \deepsec using involved but generic encodings.
We also show the practical relevance  of our approach by presenting case studies. We analyze with \tamarin key secrecy for one session of the BB84 QKD protocol for
different threat models and automatically identify minimal assumptions in terms of message authenticity and the intruder's capabilities.
We also automatically  find several attacks, some of which were well-known and others that rely on minimal security assumptions that were not previously
clearly identified.
For instance, when sufficiently many verification bits are checked,
we show that the authenticity of three specific classical messages out of four is a minimal requirement.
Namely, we show attacks when this is violated for one of the three messages and provide
a proof in our model otherwise.
In particular, we automatically found that when verification bits are not necessarily authentic,
an EPR attack completely breaks secrecy as the intruder can learn all the bits measured by Bob.
We also model the BB84 QBC protocol and automatically re-discover the EPR-based binding attack that completely defeats the protocol purpose.

Since our framework is based on well-established frameworks and verifiers, 
it can handle complex cryptographic systems containing a
quantum protocol at its core.
This can theoretically be leveraged to assess the security of a whole security system as specified,
or deployed, instead of a single quantum component in isolation.

\paragraph{Our Approach and its Abstractions.}
To achieve the above, we adopt the following abstractions and make the following modeling choices.

\looseness=-1
Based on meta-level probabilistic reasoning, we consider fixed bitstrings instead of
random bitstrings\footnote{Where those random bitstrings correspond to key bits and random bases chosen to encode
key bits}. We thus deal with {\em possibilities} rather than with probabilities.
However, even though the bitstrings under consideration are fixed,
we explore all potential executions that are possible for such bitstrings.
We believe that, by wisely choosing these bitstrings, we can capture
a wide range of logical attacks. 
However, fixing these bitstrings naively would allow the intruder to perform bitstring-dependent attacks,
which would be successful for the real system with only negligible probability.
We thus make the data of the fixed bitstrings (i) initially secret, and
(ii) partially guessable by the intruder.
The amount of data that can be guessed in our model
is defined through meta-level probabilistic reasoning.
We also identify and mitigate an ``intruder's knowledge propagation effect'':
when the intruder correctly guesses a bit $b$ picked by Alice, he automatically
learns all other bits equal to $b$ {\em for the fixed bitstrings} under consideration.
This is taken care of by modeling all elements of the bitstrings by different terms and
by adapting the equality and inequality relations accordingly.
While we will miss out many probabilistic attacks this way,
we believe this is analogous to the
gap between the symbolic and the computational approach for classical cryptography.
We recall that the goal here is to capture some {\em logical attacks}
that can be performed with non-negligible probabilities.
In short, we balance trade-offs differently:
we provide weaker guarantees in exchange for automation.

We build on the symbolic model to model a quantum channel and a quantum intruder.
We model qubits resulting from encoding of bits in orthonormal bases as an uninterpreted
function over the bit and the base.
Our viewpoint is that, analogous to classical channels in the standard symbolic model,
quantum channels should be considered to be entirely under the intruder's control.
We thus model a quantum channel where all outputs are given to the intruder
and all inputs are chosen by the intruder.
Contrary to the computational model that defines what the intruder {\em cannot do},
a symbolic model explicitly specifies what the intruder {\em can do}.
We thus choose a fixed, yet rich, set of quantum intruder capabilities:
our intruder can forward, transform, measure, and forge qubits.
Those capabilities are restricted according to relevant physics laws.
For example, a qubit is {\em consumed} (and cannot be reused) upon measurement or forwarding,
measurement can yield random data (wrong basis) or the encoded bit (matching basis), and
measurement by honest parties sometimes leaks data (\eg in EPR-attacks).


\looseness=-1
Finally, note that our framework is not complete (attacks may be missed) with regard to Nature.
This is to be expected as we only model some intruder capabilities and we then abstract them away.
More surprisingly, our framework does not provide sound falsifications (\ie no false attack) with regard to Nature either. This stems from our abstractions of random bitstrings, which might lead to attack scenarios that only happen with negligible probabilities in reality. While we took efforts to mitigate this, soundness of falsification does not currently hold and we do not see how it could with our current modeling choices.
Hence, our abstractions related to quantum capabilities are attack preserving, while the ones abstracting away probabilities are not.

\paragraph{Outline.}
We recall basics of quantum physics in Section~\ref{sec:back}.
We present some quantum protocols we would like to model in Section~\ref{sec:proto}.
We describe how we modified the regular Dolev-Yao intruder to account for probabilities
and quantum capabilities in Section~\ref{sec:DY}.
Finally, we explain how this enriched intruder model can be modeled in state-of-the art symbolic verifiers and present experimental results
in Section~\ref{sec:verif} and conclude in Section~\ref{sec:conclu}.





\section{A Mathematical Model for Quantum Physics}
\label{sec:back}
We first provide basic background in quantum physics. We closely follow the presentation of~\cite{nielsen2002quantum} and
we treat quantum objects as {\em mathematical objects} with certain properties. 
Our focus in on the general framework provided by quantum physics.
We shall work with some well-established, simple assumptions about the {\em state spaces} of the systems, their {\em dynamics}, and their {\em measurements} that mathematically specifies a general quantum information theory independent
of concrete realizations.


\subsection{States}
\setcounter{postulate}{0}
\begin{postulate}[from \cite{nielsen2002quantum}]
Associated to any isolated physical system is a complex vector space
with an inner product (that is, a Hilbert space [provided it is finite dimensional]) known as the {\em state space} of the
system. The system is completely described by its {\em state vector}, which is a unit
vector in the system's state space.
\end{postulate}
In the present paper, we assume all vector spaces to be finite dimensional and can restrict ourselves to the spaces
$\mathbb{C}^{n}$ for $n\in\mathbb{N}$.

\subsubsection{Qubits}
The state of a qubit is a unit vector in a two-dimensional, complex, vector space with an inner product such as $\mathbb{C}^{2}$. The $\ket{\cdot}$ (pronounced ``ket'') indicates that the object is a vector.
The special states $\ket 0$ and $\ket 1$ are known as the {\em computational basis} (or {\em rectilinear basis}) states\footnote{Examples of realizations of $\ket 0$ and $\ket 1$ are:
  the two different polarizations of a photon,
  the alignment of a nuclear spin in a uniform magnetic field,
  two states (ground or excited) of an electron orbiting a single atom.},
and form an orthonormal basis for this vector space.
In $\mathbb{C}^{2}$, they are defined as follows:
$$\hfill
\ket 0 \deff \mqty[1 \\ 0]
\hfill\quad\text{ and }\quad\hfill
\ket 1 \deff \mqty[0 \\ 1].
$$

A qubit can be in any {\em superposition} of $\ket 0$ and $\ket 1$:
$$
\ket \psi = \alpha\ket 0 + \beta\ket 1.
$$
However, since $\ket \psi$ must be a unitary vector, $\alpha$ and $\beta$ must satisfy: $\abs{\alpha}^2 + \abs{\beta}^2 = 1$.
Counter-intuitively, there is no way to measure precisely $\alpha$ or $\beta$ in the general case. When one measures $\ket \psi$ (with respect to the computational basis),
one obtains $0$ with probability $\abs{\alpha}^2$ ($\ket \psi$ collapse to $\ket 0$) and $1$ with probability $\abs{\beta}^2$ ($\ket \psi$ collapses to $\ket 1$). See the explanations on how the continuous variables $\alpha$ and $\beta$ can be leveraged in Section~\ref{sec:back:mecha:hidden}.

Given $\ket{\psi}$, the object $\bra{\psi}$ (pronounced ``bra'') denotes $\ket{\psi}\dag$ (\ie the Hermitian conjugate of $\ket{\psi}$).
Next, given two vectors $\ket\phi$ and $\ket\psi$, the complex number $\bra{\phi}\ket{\psi}$ denotes the inner product between $\ket \phi$ and $\ket \psi$ (\ie $\ket \phi\dag \ket\psi$).
Finally, given $\ket{\phi}\in\mathbb{C}^n$ and $\ket{\psi}\in\mathbb{C}^m$,
the object $\ket{\psi}\bra{\phi}$ denotes the linear map from $\mathbb{C}^m$ to $\mathbb{C}^n$
whose the matrix representation is
${\ket \psi}\dag\ket{\phi}\in\mathbb{C}^{n\times m}$.
Note that for a vector $\ket v\in\mathbb{C}^m$, it holds that:
$$(\ket{\psi}\bra{\phi})(\ket v)=
\ket\psi\bra{\phi}\ket{v}=
\bra{\phi}\ket{v}\ket{\psi}.$$

\subsubsection{Multiple Qubits}
We now assume that we have have two qubits available. In the classical world, the set of states
obtained by combining the two 1-bit system
is the {\em Cartesian product} of the two 1-bit system sets of states, therefore of dimension 2.
In the quantum setting, the combined system is the {\em tensor product} of the two 1-dimensional vector spaces representing
the two 1-qubit systems, and therefore is of dimension $2\times 2 =4$.
This is formally stated in the fourth postulate.
\setcounter{postulate}{3}
\begin{postulate}[from \cite{nielsen2002quantum}]
The state space of a composite physical system is the tensor product
of the state spaces of the component physical systems. Moreover, if we have
systems numbered 1 through n, and system number $i$ is prepared in the state
$\ket{\psi_i}$, then the joint state of the complete system is $\ket{\psi_1}\otimes \ket{\psi_2}\otimes\ldots\otimes\ket{\psi_n}$.
\end{postulate}
Therefore, the computational basis for pairs of qubits is:
$$
\hfill \ket{00}=\ket 0 \otimes \ket 0,\quad
\ket{01} = \ket 0\otimes\ket 1,\quad
\ket{10}=\ket 1\otimes\ket 0 \text{ and }
\ket {11}=\ket 1\otimes\ket 1.
\hfill
$$
The state vector describing two qubits is thus of the form:
$$
\ket \psi =
\alpha_{00}\ket{00} +
\alpha_{01}\ket{01} +
\alpha_{10}\ket{10} +
\alpha_{11}\ket{11},
$$
such that $\abs{\alpha_{00}}^2+\abs{\alpha_{01}}^2+\abs{\alpha_{10}}^2+\abs{\alpha_{11}}^2=1$.

\subsection{Dynamics}
\setcounter{postulate}{1}
\begin{postulate}[from \cite{nielsen2002quantum}]
  The evolution of a closed quantum system is described by a {\em unitary transformation} (\ie associated matrix is {\em unitary}).
  That is, the state $\ket{\psi}$ of the system at time $t_1$ is related
  to the state $\ket{\psi}$ of the system at time $t_2$ by a unitary operator $U$ which depends only on the times $t_1$ and $t_2$:
  $\ket{\psi'} = U\ket{\psi}$.
\label{post:dynamics}
\end{postulate}

Perhaps surprisingly, this unitarity constraint is the only constraint: any unitary transformation can be realized by
real-world systems, and conversely, any known quantum transformations can be described by a unitary transformation.
Note that a small set (called universal sets) of quantum gates (\ie atomic unitary transformations) are sufficient to generate all
possible unitary transformations.

This postulate only addresses the dynamics of completely {\em closed} systems, which is never the case in the real world.
In practice however, this postulate gives a good approximation of real life systems that are sufficiently isolated
{\em modulo noise}, but the details are outside of the scope of this paper.

A useful unitary transformation on $\mathbb{C}^2$ is the {\em Hadamard gate}:
$$
H = \frac{1}{\sqrt 2 }\mqty[1 & 1 \\ 1 &-1].
$$
$H$ transforms the rectilinear basis into the {\em Hadamard basis} (or the {\em diagonal basis}), formed by the two following qubits:
$$
\hfill
\ket{{+}} \deff H\ket{0} = \frac{1}{\sqrt{2}} \ket 0 + \frac{1}{\sqrt{2}} \ket 1
\hfill\quad\text{ and }\quad\hfill
\ket{-} \deff H\ket 1 = \frac{1}{\sqrt{2}} \ket 0 - \frac{1}{\sqrt{2}} \ket 1.
\hfill
$$


\subsection{Measurements}
\label{sec:back:measure}
\looseness=-1
Postulate~\ref{post:dynamics} describes the dynamics of closed systems only. The following postulate describes how systems evolve
when they are measured (systems are no longer closed then), and what can be observed.
\setcounter{postulate}{2}
\begin{postulate}[from \cite{nielsen2002quantum}]
Quantum measurements are described by a collection $\{M_m\}$ of
{\em measurement operators}. These are operators acting on the state space of the
system being measured. The index m refers to the measurement outcomes that
may occur in the experiment. If the state of the quantum system is $\ket{\psi}$
immediately before the measurement then the probability that result m occurs is
given by
   $p(m) \deff \expval{M_m^\dagger M_m}{\psi}$
and the state of the system after the measurement is
   $(M_m \ket{\psi})/\sqrt{p(m)}$.
The measurement operators satisfy the completeness equation [which ensure that probabilities sum to one],
   $\sum_m M_m^\dagger M_m = I$.
\end{postulate}

\begin{example}
The collection of measurement operators $\{M_0, M_1\}$, where $M_i = \ket i \bra i$, forms a quantum measurement.
It corresponds to measuring in the computational basis $\ket 0, \ket 1$.
Formally, given a qubit $\psi\deff \alpha \ket 0 + \beta \ket 1$, the probability of measuring $0$ is
$$p(0)\deff \bra \psi M_0\dag M_0\ket \psi =
\bra{\psi}\bra{0}\ket{0}\ket{0}\bra{0}\ket{\psi} =
\bra{\psi}\ket{0}\bra{0}\ket{\psi} = \abs{\alpha}^2.$$  
Similarly, $p(1)=\abs{\beta}^2$.
\label{ex:measure}
\end{example}

\subsection{Some Mechanisms and Principles}
\label{sec:back:mecha}
\subsubsection{Hidden Data}
\label{sec:back:mecha:hidden}
On the one hand, it seems that we can store an infinite amount of data with a single qubit (since there are infinitely many choices for $\alpha$ and $\beta$). On the other hand,
most of ``stored data'' is lost as soon as one wants to reads it through a measurement.
However, the essential point of quantum mechanisms is that when one applies transformations on those qubits, the continuous variables $\alpha, \beta$ do have an impact on the resulting qubits. Therefore, as long as the system is not disturbed by any observation or measurement, the dynamics keeps track of those (infinite) hidden pieces of information. This is an essential property that can be exploited to make the different alternatives that may be encoded in qubits
interfere with each other. By contrast with the classical world where alternatives exclude each other (in, \eg probabilistic programming), alternatives in the quantum world
co-exist and evolve at the same time, simultaneously. The {\em interference} mechanism can be used to effectively exploit hidden information to
outperform classical computations (see for example the super-dense coding mechanism in Section~\ref{sec:proto:coding} or Deutsch\--Jozsa algorithm presented in~\cite{nielsen2002quantum}).
Interestingly, when the number of qubits grows, the amount of hidden information one can
effectively exploits (through a measurement done after a couple of transformations) may grow exponentially
({\em parallelization} principle).

How this can be leveraged in practice is the primarily focus of {\em quantum computation}
but is unimportant to this paper.


\subsubsection{Distinguishing Quantum States}
Consider the following game: Alice and Bob agree on a set of quantum states $\{\phi_i\}$, Alice secretly chooses
some $\phi_i$ and sends it to Bob, can Bob identify the index~$i$?
As shown in Example~\ref{ex:measure}, Bob has a winning strategy if the pre-defined set of states is actually
the computational base (Bob just has to measure the unknown state with $\ket{0}\bra{0}$ for instance).
More generally, a similar reasoning shows that this still holds for any orthonormal basis.

What about other sets? Sadly, an impossibility result shows that there cannot be any quantum measurement capable of
distinguishing a state from a non-orthonormal set of states.
\begin{theorem}[from \cite{nielsen2002quantum}]
  Let $H$ be an Hilbert space.
  There is no measurement distinguishing two non-orthogonal states.
  \label{thm:dist}
\end{theorem}
\begin{proof}
See Box 2.3 on page 87 in~\cite{nielsen2002quantum}.  
\end{proof}

\begin{example}
  Consider the set $\{\ket{0},\ket{1},\ket{+}\}$. There is no quantum measurement $\{M_m\}$ that can identify
  with probability ${\ge}\frac{1}{2}$
  the state $\ket{0}$ (or the state $\ket{+}$) from the others.
\end{example}


\subsubsection{No-Cloning Theorem}
The no-cloning theorem states that there is no way to copy a qubit in the general case.
Intuitively, a transformation $U$ would be able to copy a state $\ket{\phi}$ if
it satisfied
$$U(\ket{\phi}\otimes\ket{e}) = \ket{\phi}\otimes\ket{\phi}\quad \text{(up to a global phase factor we omit)}$$
for some state $\ket{e}$.
As stated below,
there is no generic way to copy qubits that works for two different, non-orthogonal qubits.

\begin{theorem}[from \cite{nielsen2002quantum}]
  Let $H$ be an Hilbert space.
  There is no unitary transformation $U:H\mapsto H$ such that
  there exists different and non-orthogonal $\ket{\phi},\ket{\psi}\in H$ satisfying
  $U(\ket{\phi}\otimes\ket{e}) = \ket{\phi}\otimes\ket{\phi}$
  and
  $U(\ket{\psi}\otimes\ket{e}) = \ket{\psi}\otimes\ket{\psi}$.  
  \label{thm:cloning}
\end{theorem}
\begin{proof}
  Assume the existence of $U$, $\ket{e}$, $\ket{\phi}$, and $\ket{\psi}$ such that
  $U(\ket{\phi}\otimes\ket{e}) = \ket{\phi}\otimes\ket{\phi}$ and
  $U(\ket{\psi}\otimes\ket{e}) = \ket{\psi}\otimes\ket{\phi}$ hold.
  By taking the inner product of both equations, we obtain:
  $$
  ((U(\ket{\phi}\otimes\ket{e}))\dag, U(\ket{\psi}\otimes\ket{e}))
  =
  ((\ket{\phi}\otimes\ket{\phi})\dag,
  \ket{\psi}\otimes\ket{\psi}).
  $$
  Since $U\dag U=I$, we obtain
  $$
  \bra{\phi}\ket{\psi}\cdot\bra{e}\ket{e}
  =
  \bra{\phi}\ket{\psi}\cdot\bra{\phi}\ket{\psi}
  $$
  Since $\bra{e}\ket{e}=\abs{e}^1=1$, we obtain
  $\bra{\phi}\ket{\psi} = \bra{\phi}\ket{\psi}^2$ in $\mathbb{C}$. Therefore,
  either $\bra{\phi}\ket{\psi}=0$ (\ie $\ket{\phi}$ and $\ket{\psi}$
  are orthogonal) or
  $\bra{\phi}\ket{\psi}=1$ (\ie $\ket{\phi}=\ket{\psi}$).
\end{proof}
Conversely, if $\ket{\phi},\ket{\psi}\in H$ are orthogonal or equal, one can trivially build an appropriate $U$.
For the former case, one can measure the state to be copied in an orthonormal basis of $H$ whose two vectors are
$\ket{\phi}$ and $\ket{\psi}$ (using the Gram-Schmidt algorithm for instance) in order to perfectly distinguish between the two possible states given as input.
Once this is known, one can just build the appropriate qubit.
In the latter case, it suffices to let $U$ be the constant function that always returns
$\ket{\phi}\otimes\ket{\phi}$.

\subsubsection{Bell States and EPR Experiment}
\label{sec:back:mecha:epr}
\newcommand{\Bell}{\ket{\Phi^-}}
Contrary to the classical world, quantum states such as qubits do not have physical properties that exist independent of observation. Physical properties rather arise as a consequence of measurements performed upon the system. Let us see how this principle
has been experimentally validated by the {\em anti-correlations in the EPR experiment}\footnote{Historically,
  the first experiment designed to test (actually to try to invalidate) this theory
  was the EPR experiment that violates the Bell inequality.
  It has been then experimentally shown that the Bell inequality is actually not obeyed by Nature.
}.
Consider the following state, called {\em Bell state} (Appendix~\ref{ap:Bell} describes how one can obtain such a state):
$$
\Bell = \frac{\ket{01} - \ket{10}}{\sqrt{2}}.
$$
Note that such a state would not exist if the states of the union of two qubits were taken from the sum product; instead of
the tensor product (see Postulate 4). Such states that cannot be written $\ket{\phi_A}\otimes\ket{\phi_B}$
for two sub-systems $A$ and $B$ are called {\em entangled states}.

The Bell state is made of two intricated but physically disjoint qubits that can be given to Alice and Bob, who may be extremely far away from each other.
If Alice measures her own bit in the computational basis (\ie using $\ket{0}\bra{0}$ or $\ket{1}\bra{1}$),
she observes 0 with probability $\frac{1}{2}$, so does Bob. However, once one of the two has measured his own qubit,
the resulting intricated system is fixed: either $\ket{01}$ or $\ket{10}$ (global phases do not impact future observations and can be omitted). Let's suppose Alice and Bob measure almost instantly their qubits (we assume the time discrepancy between both measurements is strictly lower than $\frac{d}{c}$, where $d$ is the distance between them and $c$ is the speed of light). By symmetry, we can assume Alice measures first.
Let us explore possibilities:
\begin{itemize}
\item Alice measures
$0$ with probability $\frac{1}{2}$ and the resulting intricated system is ${\ket{01}}$. Therefore, Bob then measures $1$ with
probability 1.
\item Alice measures
$1$ with probability $\frac{1}{2}$ and the resulting intricated system is ${\ket{10}}$. Therefore, Bob then measures $0$ with
probability 1.
\end{itemize}
Bob thus always measures the opposite bit of Alice's measured bit.
This gives the ability to Alice and Bob to both flip a coin and agree on the resulting outcome, without exchanging any information.
This is something that cannot be done in the classical world.
Examples of applications of this idea are given in \Cref{sec:proto:comm,sec:proto:coding,sec:proto:tele}.




\section{Quantum Protocols}
\label{sec:proto}
We now describe how the new quantum mechanisms can be leveraged to devise new quantum protocols and notably
security protocols that are more secure than classical protocols (\Cref{sec:proto:BB84,sec:proto:comm}).
We also present two classical applications of quantum mechanisms in \Cref{sec:proto:coding,sec:proto:tele}.


\begin{figure}[t]
  \centering
\begin{verbatim}
 - Input: security parameters n, D.
 - Security: The scheme succeeds in establishing a shared key of length n with
             probability at least 1 - O(2^{-D}).
             The schemes ensures that Eve's mutual information with the  final key
             is less than 2^{-n}.

[Alice]        : Chooses s=(4+D)n random bits (d_i)
                 chooses s random bases (b_i) in {[+],[x]}
                 for all i=1..s, c_i = encoding of d_i in b_i
 Alice -> Bob  : (c_i) for i=1..s

[Bob]          : Chooses s random bases (b'_i) in {[+],[x]}
                 for all i=1..s, d'_i = measuring c_i in b'_i
 Bob -> Alice  : Done

[Alice] -> Bob : (b_i) for i=1..s

[Bob]          : Computes indices j_1,..,j_N such that b_{j_i} = b'_{j_i}.
                 On average, N is close to 2n + δn/2.
                 If N < 2n, abort the protocol.
 Bob -> Alice  : (j_i) for i=1..N

[Alice]        : Selects at random n indices k_1,..,k_n among the (j_i).
 Alice -> Bob  : (j_i), (d_{j_i}) for i=1..n

[Bob]          : If d_{j_i} != d'_{j_i} for some i, abort. Otherwise: success.

[ Not described here: Information reconciliation + Privacy amplification on the n shared
  bits to obtain m>n shared key bits. ]
\end{verbatim}
  \caption{Description of the BB84 protocol}
  \label{fig:BB84}
\end{figure}

\subsection{Quantum Key Distribution: BB84}
\label{sec:proto:BB84}
The BB84 protocol~\cite{bennett2014quantum} aims at securely distributing a session key from Alice to Bob. 
The main quantum mechanism leveraged by this protocol is the no-cloning theorem that can be used to detect
the presence of an eavesdropper.

Essentially, Alice sends $s$ qubits to Bob, each encoded either in the rectilinear basis (denoted by $[+]$) or in
the diagonal basis (denoted by $[\times]$) chosen at random. The attacker may intercept those qubits but does not learn anything: neither the value of the bits nor the bases (see Theorem~\ref{thm:dist}). Bob then measures those bits in either bases chosen at random.
On average, half those bases match Alice's choices of basis. Therefore, for those half measurements,
Bob will observe the same bits that the ones Alice has encoded and for the others, he will obtain completely random bits.
On a classical channel, they communicate indices of matching bases. Alice then use a subset of those to test
the presence of the attacker by revealing the corresponding initial secret bits. If Bob receives enough
bits that match his own measurements, both consider that enough qubits have not been eavesdropped on.
The rationale is that when the intruder eavesdropped on exchanged qubits, he cannot both measure them and keep it intact for Bob
(Theorem~\ref{thm:cloning}). Furthermore, when measuring, the intruder cannot learn both the encoded bit and the base;
that would contradict Theorem~\ref{thm:cloning} as well.
Therefore, eavesdropping on qubits introduces mismatches between verification bits sent by Alice and the ones
measured by Bob.
The shared key is then made of the matching and unrevealed bits that are shared between Alice and Bob.
\smallskip{}

A detailed presentation of the scheme is presented in Figure~\ref{fig:BB84}.

\paragraph{Attack Scenarios.}
Note that if all the channels are assumed to be under the intruder's control, the intruder
can perform a full Man-in-the-Middle (MitM) attack and defeat agreement and secrecy properties on the session key:
pretending to be Bob to Alice and pretending to be Alice to Bob. If only the \texttt{Done} is not authenticated or not replay protected in case of multiple sessions,
the attacker can still perform a MitM attack.
Both attacks are presented below.
\begin{attack}[Full Man-in-the-Middle attack]
  A MitM attacker can pretend to be Bob to Alice and Alice to Bob. He just has to emulate both roles.
  In doing so, he learns, on average, half matching $d_i$ with each honest agent.
  When he has completed both runs, he establishes a shared key with Alice and another one with Bob.
  Those shared keys are not secret while Alice and Bob think they have established a shared secret key between each other.
  
  This is a known attack that we were able to find automatically (see Section~\ref{sec:verif}).
\label{at:MiM}
\end{attack}
\begin{attack}[]
When the \texttt{Done} message is neither authentic nor replay-protected, the attacker can pretend to Alice that Bob
has received all the qubits while the attacker is still keeping them untouched. This allows him to read all the qubits
in Alice's bases and obtain all the $d_i$ that Alice has picked randomly and encoded in the qubits.
Therefore, the attacker can forge identical qubits by encoding again the same $d_i$ in Alice's bases
and sending them to Bob. The intruder can then let Alice and Bob exchange (authenticated) messages as expected by staying passive.
The final key established by Alice and Bob is completely known to the attacker since he knows all the $d_i$.
  
This is reminiscent to the Attack~\ref{at:MiM}, but we have discovered it when attempting to automatically verify the protocol
using Tamarin (see Section~\ref{sec:verif}).
\label{at:done}
\end{attack}

A different way to exploit the lack of authenticity consists in exploiting EPR attacks to learn all the bits that
Bob has measured (\ie the $d'_i$). This attack is explained below.
\begin{attack}
  We assume here that the \texttt{Done} message is authentic or another mechanism is in place to ensure that Alice
  sends her bases only after Bob has received and measured all expected qubits.
  Furthermore, we assume that all messages sent on the classical channel are authenticated except the last
  message (Alice sending verification bits). Previous attacks are no longer valid but a different attack
  still breaks secrecy as explained next.
  
  A MitM attacker could intercept the qubits sent by Alice and drop them.
  Instead, the attacker could forge as much EPR pairs as qubits Alice sent, and send
  one share of each pair to Bob.
  The intruder then lets the message containing Alice's bases go through and then eavesdropps on Bob's bases.
  He can now measure his half of each EPR in Bob's basis and therefore obtain the same bit that Bob has obtained.
  At this point, the intruder knows all Bob's bases and all the bits Bob has measured\footnote{Note that, even in Attack~\ref{at:MiM} or Attack~\ref{at:done}, the attacker was not able to gain as much information.}.
  Therefore, he can easily choose appropriate verification bits and send those to Bob who will think he has
  established a secret key with Alice while it is known to the intruder.

  We have discovered the attack using Tamarin (see Section~\ref{sec:verif}).  
\label{at:EPR}
\end{attack}

We think that assuming the quantum channel to be authentic is unrealistic.
The same applies to the classical channel.
We think that authenticity should be provided by cryptography. For instance, signatures could be used
to sign all messages sent over classical channels. Even if the signature primitive is not quantum-resistant,
the resulting scheme would still provide unconditional perfect forward secrecy. Which is still better than
all classical schemes, excluding impractical OTP and the like.
The analysis we eventually conduct allows us to identify minimal authenticity requirements.

\begin{figure}[t]
  \centering
\begin{verbatim}
 - Input: a bit b Alice wants to commit on, security parameters n.
 - Security: Bob or Eve do not learn b before Alice unveil b.
             If Bob thinks Alice has commited on b' at the end of the unveiling
             procedure then b=b'.

----- Commit Procedure
[Alice]        : Chooses n random bits (d_i)
                 base=[+] if b=0 and base=[X] if b=1
                 for all i=1..n, c_i = encoding of d_i in base
 Alice -> Bob  : (c_i) for i=1..n

[Bob]          : Chooses s random bases (b'_i) in {[+],[x]}
                 for all i=1..n, d'_i = measuring c_i in b'_i

----- Unveil Procedure
[Alice] -> Bob : b, base, (d_i) for i=1..n

[Bob]          : Computes indices j_1,..,j_N such that b'_{j_i} = base
                 On average, N is close to n/2
                 If N < threshold, abort the protocol.
                 If d'_{j_i}=d_{j_i} for i=1..N then Bob assumes Alice has commited
                 on b.
\end{verbatim}
  \caption{Description of the bit commitment protocol derived from the BB84 protocol}
  \label{fig:commitment}
\end{figure}

\subsection{Bit-commitment Protocol}
\label{sec:proto:comm}
The BB84 scheme can be adapted into a bit commitment protocol~\cite{bruss2007quantum}.
Essentially, instead of randomly choosing bases to encode qubits, Alice chooses an uniform basis
that depends on the bit she wants to commit on.
Bob still measures in random bases.
When Alice wants to unveil the bit and prove her commitment, she publishes the base she used with the bits she
encoded in the qubits. Bob can check that, for the bases that match Alice's ones, sufficiently many measured bits match with 
the bits Alice pretended to have encoded.
\smallskip{}

A detailed presentation of the scheme is presented in Figure~\ref{fig:commitment}.
Contrary to BB84 QKD, authenticity of messages is not required.

\paragraph{Attack.}
A well known attack defeating the main security goal of the protocol exploits EPR pairs forged by Alice to cheat on Bob.
\begin{attack}
  An intruder could convince Bob that he has committed on a bit $b$ while he has not actually committed on any value.
  To do so, the intruder can proceed as follows:
  \begin{enumerate}
  \item The intruder creates $n$ EPR pairs and sends one half of each to Bob. Hence, the intruder does not commit on any value.
  \item Bob measures those qubits in random bases $(b_i)_i$ and obtain bits $(d_i)_i$.
  \item If the Intruder would like unveil a value $b$ and convince Bob he has comitted on it at step 1. (before sending the qubits), then he measures
    all his shares in the appropriate base ($\mathrm{base}=[+]$ when $b=0$ and $\mathrm{base}=[\times]$ when $b=1$) and obtains the bits $(d'_i)_i$.
    The intruder then sends to Bob the values $b$, $\mathrm{base}$, and $(\lnot d'_i)_i$, where $\lnot$ denotes the negation.
  \item For the bases $b_i=\mathrm{base}$, it holds that $d_i=\lnot d'_i$ (see Section~\ref{sec:back:mecha:epr}). Therefore, Bob assumes that the Intruder has committed on $b$ at step 1.
  \end{enumerate}
  This is a known attack that we were able to find automatically (see Section~\ref{sec:verif}).
\label{at:EPR-bit}
\end{attack}

\subsection{Other Quantum Protocols}
\label{sec:proto:other}
We now briefly present other important applications of quantum mechanisms that may be useful to better understand those mechanisms but that are unimportant for the rest of the paper.

\subsubsection{Super Dense Coding}
\label{sec:proto:coding}
{\em Setup}: same setup as the EPR experiment (see Section~\ref{sec:back:mecha}). So Alice and Bob have their own qubit that are intricated in a Bell state $\ket{\Psi^-}$.

\noindent{\em Goal}: Alice wants to send two bits of information but can send only one qubit.

\noindent{\em Idea}: depending on $b_0,b_1\in\mathbb{B}$, Alice transforms the Bell state $\ket{\Phi^-}$ into one of the four
states in the Bell basis (see Appendix~\ref{ap:Bell}). She can do that by only interacting with her half of the EPR pair.
Alice then sends her qubit to Bob who can measure the two intricated qubits
in the Bell basis and determine with probability 1 what was the intricated state and therefore the bits $b_0,b_1$.

\subsubsection{Quantum Teleportation}
\label{sec:proto:tele}
{\em Setup}: same setup as the EPR experiment (see Section~\ref{sec:back:mecha}).

\noindent{\em Goal}: Alice wants to send a qubit $\ket{\phi}$ to Bob but can only send classical information to Bob.
Sending a qubit means that Bob should have an identical qubit (with exactly the same continuous variables) at his disposal.

\noindent{\em Idea}: Alice will interact her qubit $\ket{\phi}$ with the qubit corresponding to her half of the EPR pair. Alice then measures both qubits at her disposal (\ie $\ket{\phi}$ and her share of the EPR pair) and sends the results to Bob. Based on those two classical bits, Bob will perform one of four appropriate transformations on his half of the EPR pair. By appropriately choosing the transformations made by Alice and Bob, it is possible to
make so that the resulting qubit Bob obtains at the end is exactly $\ket{\phi}$; \ie the initial qubit Alice wanted to send to Bob.
Note that, in the process, Alice loses her qubit $\ket{\phi}$ so this does not violate the no-cloning principle.


\section{Quantum Dolev-Yao Intruder}
\label{sec:DY}
After recalling basic notions about the regular, classical Dolev-Yao attacker and
symbolic models (Section~\ref{sec:DY:regular}), we explain how we extend it with
probabilities and intruder's guessing capabilities (Section~\ref{sec:DY:proba}),
and with intruder's quantum capabilities and control over quantum channels (Section~\ref{sec:DY:qu}).

\begin{figure}[tb]
  \centering
  \includegraphics[width=1\textwidth,page=1]{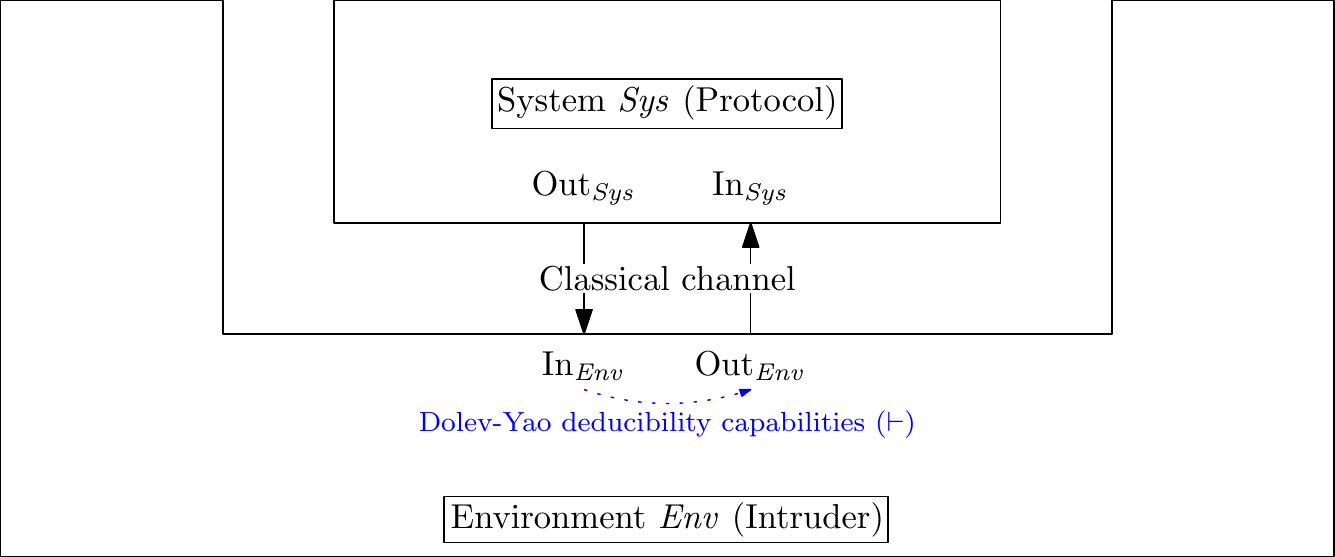}
  \caption{Regular Dolev-Yao Intruder}
  \label{fig:dolevYao:regular}
\end{figure}%

\subsection{Symbolic Model and Classical Dolev-Yao Intruder}
\label{sec:DY:regular}
\newcommand{\terms}{\mathcal{T}_\Sigma(\mathcal{X}\cup\mathcal{N})}
Figure~\ref{fig:dolevYao:regular} depicts a high-level representation of the
usual Dolev-Yao attacker who has complete
control over the classical channels (he learns all outputs and chooses all inputs)
and has deduction capabilities (he can compute new messages from the messages he already knows).

Exchanged messages are represented by terms in a pre-defined term algebra $\terms$,
that is a $\Sigma$-algebra on a set of variables $\mathcal X$ and names $\mathcal{N}$
for a pre-defined set of function symbols $\Sigma$.
A semantics can be given to those symbols through equational theories or reduction rules.
The deduction capabilities can then be expressed in a formal way,
for example through deductive systems such as natural deduction.
For instance, Figure~\ref{fig:DY-nat} 
depicts intruder's deductive capabilities for symmetric encryption and pairing
whose functions symbols are
$\Sigma=\{\senc(\cdot,\cdot),\pair{\cdot}{\cdot}\}$.
Note that function symbols of arity 0 are public constants.

\begin{figure}[tbh]
  \centering
$  \begin{array}{ccc}
  \infer[]{\Gamma\vdash m}{\Gamma\vdash\senc(m,k)&\Gamma\vdash k}\quad\quad
  \infer[i\in\{1,2\}]{\Gamma\vdash m_i}{\Gamma\vdash \pair{m_1}{m_2}}\\[10pt]
  \infer[f\in\Sigma]{f(m_1,\ldots,m_n)}{\Gamma\vdash m_1,\ldots, \Gamma\vdash m_n}\quad\quad
  \infer[]{\Gamma\vdash m}{m\in\Gamma} \quad\quad
  \end{array}$

  \caption{Example of deductive system for representing Dolev-Yao intruder's deduction
    capabilities.
    The intruder's knowledge is given by the set of terms $\Gamma$.
    $\Gamma\vdash m$ denotes that the intruder can deduce $m$ from $\Gamma$ and a proof
    thereof is a {\em deduction tree} with $\Gamma\vdash m$ as a conclusion (at the bottom).
  }
  \label{fig:DY-nat}
\end{figure}

A formal semantics can be given to security protocols by considering the interleaving semantics
for the different agents of the protocol with the intruder, communicating through classical channels
(see Figure~\ref{fig:dolevYao:regular}) and starting with an initial intruder's knowledge $\Gamma_0$
containing public values. In such a model, the intruder learns all the protocol outputs
and can choose all protocol inputs provided that he can deduce such terms (\ie he can choose
the term $M$ as input, provided that $\Gamma\vdash M$ for $\Gamma$ the union of the set of previous protocol outputs
and $\Gamma_0$).

In the subsequent sections, we introduce additional intruder capabilities:
guessing capabilities and a constrained control over the quantum channel.
The resulting quantum Dolev-Yao attacker is depicted in Figure~\ref{fig:dolevYao:quantum}.
Our goal with this quantum symbolic attacker is to model quantum protocols
such as QKD and quantum bit commitment protocols precisely enough to capture
interesting attacks such as those given in \Cref{sec:proto:BB84,sec:proto:comm}
but abstractly enough to fully automatically verify them.
\bigskip{}

\subsection{Probabilities and Guessing Capabilities}
\label{sec:DY:proba}
Symbolic models and verification methods based thereon usually do not account for
probabilities. For instance, the creation of a nonce $R$ (a large random number) by a certain role
is usually modeled as the creation of a fresh free name $n\in\mathcal{N}$.
Moreover, the security parameters such as key lengths
are usually abstracted away since those notions have no counterpart in a term algebra.

Such known limitations of symbolic modeling and verification raise two important challenges
when modeling quantum protocols such as QKD or QBC:
\begin{enumerate}
\item Quantum protocols are essentially {\em probabilistic} distributed programs.
  The protocol logic itself includes probabilistic reasoning.
  For instance for BB84 QKD, Bob picks random bases
  and then aborts if there is not enough bases matching Alice's ones.
  This is to be contrasted with most of classical security protocols that are essentially
  non-probabilistic distributed programs based on probabilistic cryptographic primitives,
  that are abstracted away in the symbolic model anyway.
  Furthermore, contrary to classical security protocols, the security parameters in quantum protocols
  not only impact the cryptographic messages, but also impact the protocol logic (\eg number of qubits\footnote{We shall
    see that qubits need to be modeled as separate terms in the symbolic model in order to capture algebraic
    properties of transformations on them.}
  Alice shall send in BB84).
  Finally, those parameters also impact the security goals (\eg number of shared bits that are secret).
  How can symbolic models take such probabilistic protocol logic and security parameters into account?

\item Most messages that are exchanged and manipulated are random bits.
  On the one hand, those values are chosen randomly and are supposed to be private, not known to the attacker.
  On the other hand, the attacker is expected to correctly guess some of them, actually half of them on average.
  For instance, in BB84 QKD, the event for which the attacker guesses half of the bits $b_i$ happens with
  probability $\frac{1}{2}$.
  This intruder capability of guessing some bits is crucial for not missing basic attacks
  such as MitM attacks (\eg Attack~\ref{at:MiM}). Indeed, since the different honest roles are essentially probabilistic
  programs, the attacker needs such capabilities to emulate them.
  How can the symbolic model and the Dolev-Yao intruder be adapted to account for such probabilistic guessing
  capabilities?
\end{enumerate}

We address both problems using abstractions and approximations we describe in the following sections.

\subsubsection{Exploring Possibilities rather than Probabilities}
Instead of reasoning with probabilities by taking into account probability distributions
of the different events in the protocol semantics,
we consider {\em possibilities} via fixed samples in the sample space, that could happen with high probability.
Similarly, our modeling is not parametric in the security parameters (\eg $n$ in BB84 QKD).
We only model and analyze the protocol for fixed, small values of those parameters.

For simplicity, we assume in the following that sample sets are in bijection with $\mathbb{B}$,
such as bits $(d_i)_i$ in $\mathbb{B}$ or bases $(b_i)_i$ in $\{[+],[\times]\}$. All values
from those sample sets are taken from public terms, \ie function symbols of arity 0.

\begin{abstraction}
  Our modeling and analyses consider fixed security parameters and bitstring lengths.
  Security goals are instantiated for those fixed security parameters.
  Bitstrings picked at random by honest agents are replaced by fixed bitstrings
  that correspond to highly probable events.
  We typically consider bitstrings with the same amount of zeros and ones.
  \label{abs:fix}
\end{abstraction}

Even though it seems that we only consider fixed scenarios, we still consider an adversarial
environment. Therefore, from the parameters and samples initially fixed, a multitude of (often infinitely many)
executions are still possible. In particular, the attacker can actively deviate from the expected,
honest execution.

\begin{example}
  For the BB84 protocol, we may be interested in the case $n=4$
  for which two of Alice's and Bob's bases match.
  This leads to one verification bit and one key bit, which is enough
  to explore many interesting execution traces of the protocol
  and to capture the aforementioned attacks.
  A concrete scenario along those lines could be:
\begin{verbatim}
------------ Fixed Samples -----------------
  Alice chooses      : b1 =[+], b2 =[x], b3 =[+], b4=[x]
                     : d1=0, d2=1, d3=1, d4=0
  Bob chooses        : b1'=[+], b2'=[+], b3'=[x], b4'=[x]
------------ In Honest Execution:-----------
  Matching bases     : b1, b4 in honest execution
  Verification bases : second of matching base, (b4 in honest execution)
  Final key bits     : remaining of matching bases, (d1 in honest execution)
\end{verbatim}
  The bits $(d_i)$ are arbitrarily fixed but the particular values we have chosen here
  do not matter because, due to other abstractions we describe later,
  we will eventually model them as pairwise indistinguishable terms (see Section~\ref{sec:DY:proba:name}).
  In contrast, the values chosen for $(b_i)$ and $(b_i')$ will have an impact. We have chosen here
  high-probability bitstrings (\eg there are two matching bases on average).
  
  Note that the verification base and the final key bit is not fixed since they depend on the adversarial environment:
  Alice's and Bob's view on the execution that is currently happening may not match.
  For instance, an attacker could tamper with the bases sent by Bob to Alice in order to make Alice
  compute different matching bases.
\label{ex:fix}
\end{example}

\subsubsection{Random Bits as Annotated Private Names}
\label{sec:DY:proba:name}
\newcommand{\sample}{\mathcal{S}}
\newcommand{\sampleN}{\mathcal{S}_\mathcal{N}}
\newcommand{\samples}{\mathbb{S}_\mathcal{N}}
\newcommand{\model}[1]{\llbracket #1 \rrbracket}
\looseness=-1
The previous design choice (Abstraction~\ref{abs:fix}) introduces an extra issue: as the fixed bitstrings are taken
from public values (\eg $(d_i)_i\in\mathbb{B}^n$ for BB84 QKD),
then the attacker can ``guess'' all those bits with probability~1.
This is not surprising as agents no longer pick those bits at random, allowing the attacker to adopt bitstring-specific
winning strategies.
More formally, the problem is that such fixed bitstrings are already in the initial intruder's knowledge $\Gamma_{0}$.
To remove this unrealistic attacker's capability, we shall take fixed bitstrings from private
values outside of $\Gamma_0$, \eg names in $\N$. However, for reasons that will become clear later, we shall also
keep track of the role that picked a certain bit,
the bitstring from which it was taken, its position in the bitstring, and the value we have chosen for it.

\begin{definition}[Fixed bitstrings modeling]
  We model a bit that is picked at random
  by an honest party $\mathrm{role}$ by a term
  $\bit(\mathtt{seed}, \mathrm{bitstring}, \mathrm{position}, \mathrm{role}, \mathrm{value})$
  where:
  \begin{itemize}
  \item $\bit(\cdot,\cdot,\cdot,\cdot,\cdot)\in\Sigma$ is a new function symbol,
  \item $\mathtt{seed}\in \N$ is a new name (hence secret) uniform for all bits and bitstrings,
  \item $\mathrm{bitstring}$ is a public constant that denotes the bitstring from which this bit is taken,
  \item $\mathrm{position}$ is a public constant in $\mathbb{N}$ that denotes the position of this bit in the underlying bitstring,
  \item $\mathrm{value}$ is a public constant that denotes the value of the bit in the fixed bitstring under consideration.  
  \end{itemize}
  \label{def:name}
\end{definition}

\begin{example}[Continuing Example~\ref{ex:fix}]
  We assume that $(b_i),(b'_i),(d_i)$ are as defined in Example~\ref{ex:fix} and we associate
  to them the terms $\model{\cdot}$ defined below:
  \begin{itemize}
  \item $\model{b_i} = \bit(\mathtt{seed},\mathtt{b},i,\mathtt{Alice},b_i)$ for $1\le i\le 4$,
  \item $\model{d_i} = \bit(\mathtt{seed},\mathtt{d},i,\mathtt{Alice},b_i)$ for $1\le i\le 4$,
  \item $\model{b'_i} = \bit(\mathtt{seed},\mathtt{b},i,\mathtt{Bob},b_i)$ for $1\le i\le 4$.
  \end{itemize}
  \label{ex:name}
\end{example}

As $\mathtt{seed}$ is never used elsewhere, it always remains secret.
Therefore, except if it is output,
a bit $\bit(\mathtt{seed}, \mathrm{bitstring}, \mathrm{position}, \mathrm{role}, \mathrm{value})$
cannot be deduced by the attacker.
We shall see that the other arguments of $\bit()$ can be leveraged 
to define a meaningful equality relation over bits and 
to give the attacker some guessing capabilities (Section~\ref{sec:DY:proba:proba}).


\paragraph{Relations between bits.}
By default, two bits as modeled above are always different except when both the underlying bitstrings, positions, roles, and values are equal.
For instance, we would never have $\model{b_i}\theo \model{b_i'}$, even when $b_i=b_i'$.
This would prevent even basic protocols from being executable, including QKD BB84 as there would be no matching bases.
We should remedy this problem by interpreting the function symbol
$\bit()$ in a clever way. This has to be done carefully in order to satisfy the following requirements:
\begin{enumerate}
\item \looseness=-1
  We shall not consider that bits can be equal directly in the term algebra\footnote{Obviously excluding syntactically equal terms.}. Indeed, this would yield an unwanted
  {\em intruder's knowledge propagation effect}:
  \eg when one considers $\model{d_1}\approx_E \model{d_4}$ when it happens that $d_1=d_4$ for the fixed bitstrings under consideration,
  then one assumes that an attacker learning $\model{d_1}$, immediately learns $\model{d_4}$ as a consequence.
  This is unwanted because $d_1=d_4$ might hold for the fixed bitstrings, but this happens with probability
  $\frac{1}{2}$ without our abstractions, and negligible probability when repeated on sufficiently many bits $d_i$s.
\item Instead of considering that some bits might be equal, we consider them to be {\em interchangeable}. This is done by basing all equality and disequality tests performed by honest parties (and quantum transformations as we shall see) on a $\Sigma$-congruence relation $\approx^b_E$ that is coarser than $\approx_E$,
  \ie  $\approx^b_E \supseteq \approx_E$.
  Two bits, \eg $\model{b_1}$ and $\model{b_1'}$, are made {\em interchangeable} by letting
  $\model{b_1} \approx^b_E \model{b_1'}$.
  Indeed, all honest parties will behave the same when given $\model{b_1}$ or $\model{b_2}$.

  However, for the attacker, and for the security goals we eventually consider,
  equality is considered modulo $\theo$. Hence, breaking secrecy of $\model{d_1}$ does not immediately implies
  breaking secrecy of $\model{d_4}$, even if $\model{d_1}\approx_E^b \model{d_4}$.
\item Still, making too many bits {\em interchangeable} might introduce spurious attacks where an attacker feeds an agent who expects a bit $\model{b'}$ with
  a bit $\model{b}\approx^b_E \model{b'}$ he knows, while $b$ never equals $b'$ in reality. We thus seek to minimize the relation $\approx^b_E$ as explained next.
\end{enumerate}

\begin{definition}[$\approx^b_E$]
  The $\Sigma$-congruence relation $\approx^b_E \supseteq \approx_E$ is protocol dependent but follows straightforward rules:
  \begin{enumerate}
  \item it always holds that:
    $\bit(\mathtt{seed}, \mathrm{bits}, \mathrm{pos}, \mathrm{role}_1, \mathrm{v})
    \approx^b_E
    \bit(\mathtt{seed}, \mathrm{bits}, \mathrm{pos}, \mathrm{role}_2, \mathrm{v})$.
  \item when an honest party compares different bits from the same bitstring $\mathrm{bits}$ (\eg the bases $(b_i)_i$ for BB84 QBC),
    we extend $\approx^b_E$ as follows:
    $\bit(\mathtt{seed}, \mathrm{bits}, \mathrm{pos}_1, \mathrm{role}_1, \mathrm{v})
    \approx^b_E\linebreak[4]
    \bit(\mathtt{seed}, \mathrm{bits}, \mathrm{pos}_2, \mathrm{role}_2, \mathrm{v})$.
  \end{enumerate}
  Other extensions are possible, depending on the protocol.
  \label{def:name-equ}
\end{definition}

\begin{example}[Continuing Example~\ref{ex:name}]
  We consider that for any $\mathrm{bitstring}\in\{b,d\}$,
  $\mathrm{position}\in[1,4]$,
  $\mathrm{value}\in\{1,2\}$, the following bits are interchangeable:\\[5pt]
  \null\hfill
  $\bit(\mathtt{seed}, \mathrm{bitstring}, \mathrm{position}, \mathtt{Alice}, \mathrm{value})
  \approx^b_E
  \bit(\mathtt{seed}, \mathrm{bitstring}, \mathrm{position}, \mathtt{Bob}, \mathrm{value})$.
  \hfill\null
  \label{ex:name-equ}
\end{example}

\begin{abstraction}
  We consider a new function symbol $\bit(\cdot,\cdot,\cdot,\cdot,\cdot)\in\Sigma$ and a fresh
  name $\mathtt{seed}\in \N$ and use it to model bitstrings as described in Definition~\ref{def:name}.
  We consider an equality relation $\approx^b_E \supseteq \approx_E$ over $\terms$ as specified in
  Definition~\ref{def:name-equ}.
  Any (dis)equality test carried out by honest agents of the system
  are performed modulo $\approx^b_E$, instead of modulo $\approx_E$ as standard.

  
  \label{abs:name}
\end{abstraction}

\subsubsection{Guessing Capabilities}
\label{sec:DY:proba:proba}
We first note that if an attacker is able to deduce a bit associated to a value, he must be able to deduce
the bit associated to the dual value as a direct consequence. This corresponds to an intruder's behavior that is successful
with probability 1 without our abstractions.
Therefore, we enrich the deductive system under consideration with the following additional rule:
\begin{equation}
\infer[\textsc{Complem}]{\Gamma \vdash \bit(\mathtt{seed}, \mathrm{bits}, \mathrm{pos}, \mathrm{role}, \mathrm{v})}{\Gamma\vdash \bit(\mathtt{seed}, \mathrm{bits}, \mathrm{pos}, \mathrm{role}, \mathrm{v'}) & \Gamma \vdash v}
\raisebox{6.2pt}{.}
\label{eq:complem}
\end{equation}

Furthermore, as mentioned earlier, the attacker must have some guessing capabilities, at least
so that he has similar capabilities as the honest roles, who can notably randomly pick a series of bits (\eg bases $b'_i$)
and obtain half correct guesses on average (\eg matching bases $b_i=b'_i$). 
We now would like to give back to the attacker some guessing capabilities
we got rid of by using private terms to model bits picked at random (Abstraction~\ref{abs:name}).

We first consider the following, additional rule in our deductive system
that allows the attacker to guess a bit that models part of
the fixed bitstrings under consideration (\eg $\model{b_1}$ in Example~\ref{ex:name}).
\begin{equation}
\infer[\textsc{Guess}(\mathrm{bits},\mathrm{role},\mathrm{pos})]
{\Gamma \vdash \bit(\mathtt{seed}, \mathrm{bits}, \mathrm{pos}, \mathrm{role}, \mathrm{v})}{\Gamma\vdash v}
\label{eq:guess}
\end{equation}

In order to limit this new unconstrained intruder's guessing capability, we put restrictions on
valid {\em deduction trees} (see Figure~\ref{fig:DY-nat}) that forbid some uses of the rule $\textsc{Guess}(\cdot,\cdot,\cdot)$.
The goal here is to get rid of some of the attacker's behaviors that exploit
the fact that only fixed bitstrings are considered using our abstractions
but that would have a negligible probability of success without
our abstractions and for sufficiently large security parameters.
For instance,  the attacker (blindly) guessing the correct base $b_i$, happens with probability $\frac{1}{2}$ but
with probability $2^{-n}$ for n bases; therefore it is acceptable to limit our analyses to attacker's behaviors
that guess at most $\frac{n}{2}$ bases (this happens with probability $\frac{1}{2}$ as shown
in Remark~\ref{rmk:proba} below).
Those restrictions are then globally lifted to executions.

\begin{abstraction}
We enrich the deductive system under consideration with the rules~\ref{eq:complem},\ref{eq:guess}.
Furthermore, we consider a set of restrictions on the use of $\textsc{Guess}(\cdot,\cdot,\cdot)$.
Those restrictions are then globally lifted to executions:
executions that are considered should have deduction trees that globally respect those
constraints.
Such constraints can be derived, on a case-by-case basis, by analyzing the probability events
in the protocol (see Remark~\ref{rmk:proba} below).
Note that those constraints can be added gradually when spurious attacks
(\ie attack traces that happen with negligible probability without our abstractions) are found,
hence gradually refining our abstractions.

We are notably interested in constraining the number of guesses of bits in
a single bitstring.
Constraints should at least forbid the use of more than $\frac{n}{2}$ of the
$\textsc{Guess}(\mathrm{bits},\mathrm{role},\mathrm{pos})$ rules for some fixed $\mathrm{bits}$ and $\mathrm{role}$. 
Furthermore, for bits that are equal for the fixed bitstrings under consideration,
we usually need to add constraints that limit further the guessing capabilities of bits whose the position
satisfies certain properties, \eg position at which two bits at this position, possibly from other bitstrings, are equal.
Typically, we should take into account the probability of correctly guessing bits having this property.
The rational is again to get rid of bitstring-dependent attacker winning strategies.
Guidance can be obtained by computing the probabilities of such events as shown next.
\label{abs:guess}
\end{abstraction}

\begin{remark}
  \label{rmk:proba}
  More formally, constraints can be derived by interpreting the probability of certain events in
  the following probability space:
  $\Omega=\mathbb{B}^n$ (bitstrings whose length is the security parameter $n$),
  $\mathcal{F}=2^\Omega$ (sets of bitstrings),
  $P(E)$ describes the probability of having $b\in E$ when taking $b$ uniformly randomly from $\Omega$.
  We would typically express the probability of ``guessing'' by taking a target bitstring $B\in\Omega$
  corresponding to the fixed bitstring under consideration, with equally many zeros and ones.
  For instance, we would use
  $$
  E_{\ge k}=\{b\in\Omega\ |\ \exists i_1,\ldots,i_k, \forall j\in[1,k], b_{i_j}=B_{i_j}\}
  $$
  to denote the event of guessing at least $k$ bits.
  We can show that
  $P(E_{\ge n/2})=\frac{1}{2}$ (non negligible) while, for any positive, linear function $f$,
  $P(E_{\ge n/2+f(n)})\xrightarrow[n\mapsto \infty]{}0$ (negligible). That is why we consider that at most half the bits
  of a bitstring can be guessed.

  Similarly for $F_{\ge k}=\{b\in\Omega\ |\ \exists i_1,\ldots,i_k, \forall j\in[1,k], b_{i_j}=B_{i_j}=B_{i_1}\}$,
  one has that 
  $P(F_{\ge n/4})=\frac{1}{2}$ (non-negligible) and
  $P(F_{\ge n/4+f(n)})\xrightarrow[n\mapsto \infty]{}0$ (negligible)
  that is why we consider that only a fourth of the bits having the same value can be guessed.
\end{remark}
\begin{proof}
  For the event $E$ ($E_{=k}$ is defined straightforwardly), one has:
\begin{equation*}
  \displaystyle
  \begin{split}
 E_{\ge {\frac{n}{2}}} & = \sum_{k=\frac{n}{2}}^{k=n} P(E_{=k}) \\
 & = \left(\sum_{k=\frac{n}{2}}^{k=n} C(n,k)\right)\cdot 2^{-n} \\
 & = \frac{2^n}{2}\cdot 2^{-n} \\
 & = \frac{1}{2} \\
\end{split}
\end{equation*}
We do not detail the computation of the negligible probabilities, but the idea is that one would have to subtract
a linear number of terms of the form $2^{-n}\cdot C(n,\frac{n}{2}+j)$ for $0\le j< f(n)$. The sum of a linear number of such terms
(\eg $\frac{n}{4}$) has a limit of $\frac{1}{2}$ when $n$ approaches $+\infty$.
For the event $F$ ($F_{=k}$ is defined straightforwardly), one has:
\begin{equation*}
\begin{split}
 F_{\ge \frac{n}{4}} & = \sum_{k=\frac{n}{4}}^{k=n} P(F_{=k})\\
 & = \sum_{k=\frac{n}{4}}^{k=\frac{n}{2}} P(F_{=k})\\
 & = \left( \sum_{k=\frac{n}{4}}^{k=\frac{n}{2}} C(\frac{n}{2},k)\cdot 2^{\frac{n}{2}}\right)\cdot 2^{-n} \\
 & = \left(\frac{2^{\frac{n}{2}}}{2}\cdot 2^{\frac{n}{2}}\right)\cdot 2^{-n} \\
 & = \frac{1}{2} \\
\end{split}
\end{equation*}
We do not detail the negligible probability computation either.
\end{proof}

\begin{example}[Continuing Example~\ref{ex:name}]
  We constrain the intruder's guessing capabilities by taking the conjunction of the following:
  \begin{enumerate}
  \item At most 2 instances of the rules $\textsc{Guess}(\mathtt{b},\mathtt{Alice},i)$, $i\in[1,4]$.
  \item At most 2 instances of the rules $\textsc{Guess}(\mathtt{b},\mathtt{Bob},i)$, $i\in[1,4]$.
  \item At most 2 instances of the rules $\textsc{Guess}(\mathtt{d},\mathtt{Alice},i)$, $i\in[1,4]$.
  \item At most 1 instance of the rules $\textsc{Guess}(\mathrm{bits},\mathrm{role},i)$, $\mathrm{role}\in\{\mathtt{Alice},\mathtt{Bob}\}$, $\mathrm{bits}\in\{b,d\}$, $i\in\{1,4\}$.
    Indeed, it holds that $b_1=b_4$.
  \item At most 1 instance of the rules $\textsc{Guess}(\mathrm{bits},\mathrm{role},i)$, $\mathrm{role}\in\{\mathtt{Alice},\mathtt{Bob}\}$, $\mathrm{bits}\in\{b,d\}$, $i\in\{2,3\}$.
    Indeed, it holds that $b_2=b_3$.
  \end{enumerate}
  The first three constraints reflect the fact that, on average, the intruder cannot guess more than
  half the bits in a bitstring.
  The last two conditions express the conditional probability that the attacker can guess
  the base or the bit $d_i$ whose index corresponds to a matching base.

  Note that for an execution to be explored in our analysis,
  those constraints should be satisfied not only for each single deduction tree occurring
  in the execution but also globally when considering the combination of all deduction trees
  occurring in the execution.
\label{ex:guess}
\end{example}

\subsection{Control over Quantum Channels}
\label{sec:DY:qu}
Our intruder's refinement has now some guessing capabilities.
If quantum channels were supposed to be completely secure, there would be no need
to extend the intruder's knowledge further.
However, such a weak threat model would be at odds with the extremely strong security guarantees
such quantum protocols aim at establishing (\eg unconditional secrecy).
We think that, very much like for classical channels for the regular Dolev-Yao attacker,
quantum channels should be considered as being entirely under the intruder control.
We shall thus explain how the attacker can manipulate, transform, measure, and forge qubits.
Our ambitions are quite modest here: we specify a generic intruder that is powerful enough
to perform aforementioned attacks (\Cref{at:EPR,at:EPR-bit,at:MiM,at:done})
instead of a full quantum-capable intruder, which would
require accounting for a comprehensive mathematical quantum model.

\subsubsection{Qubits and Quantum Channels}
For now, we abstract away the infinite set of possible qubits and only consider states that
are basis vectors for the orthonormal bases considered in the protocol under scrutiny (\eg $[+]$ and $[\times]$ for BB84).
Therefore, we model qubits using a new function symbol $\qubit(\cdot,\cdot)$ of arity two.
The term $\qubit(d,b)$ denotes the basis vector $d$ in the basis $b$.
\begin{example}[Continuing Example~\ref{ex:guess}]
  For the fixed bitstrings we consider here, the term $\qubit(\model{d_1},\model{b_1})$ models the qubit which is the encoding
  of the classical bit $d_1=0$ encoded in the basis $b_1=[+]$ that Alice sends to Bob.
 \label{ex:qubit}
\end{example}
\begin{abstraction}
  Qubits are modeled by the function symbol $\qubit(\cdot,\cdot)$.
  The term $\qubit(d,b)$, where $d,b$ are two terms of the form $\bit(\cdot,\cdot,\cdot,\cdot,\cdot)$, 
  models the qubit corresponding to the classical bit $d$ encoded in the basis corresponding to the bit $b$.

  When an honest agent measures a qubit $\qubit(d,b)$ with regard to the base $b'$, he obtains $d$ if $b=b'$
  and a fresh name from $\mathcal{N}$ otherwise.
  \label{abs:qubit}
\end{abstraction}

When sent, qubits are given to the attacker. The attacker can choose what are the qubits that are received by
the protocol's agents provided that he can deduce them from his knowledge.
We extend the deductive system to take this into account:
judgments are now of the form $\Gamma; \Delta \vdash M$ where $\Gamma$ contains the previous classical outputs
and $\Delta$ contains the previous quantum outputs along with unique identifiers that
are always pairwise different.
For instance, $(\qubit(\model{d_1},\model{b_1}), \mathrm{out}_1)$ could be an element of $\Delta$.
The purpose of the identifiers is explained later.

\newcommand{\vdashQ}{\vdash\hspace*{-5pt}_\textsf{Q}\hspace{3pt}}
We introduce a new deductive system $\vdashQ$ that formally defines how the attacker can produce qubits,
whose judgments are of the form $\Gamma; \Delta \vdashQ M$.
In order to let the attacker directly use past quantum outputs, we extend the deductive system with:
$$
\infer[\textsc{ID}_\textsf{Q}(q,\mathrm{id})]{\Gamma; \Delta \vdashQ q}{(q,\mathrm{id}) \in \Delta}
\raisebox{6.2pt}{.}
$$
However, with such a rule, the attacker would be able to copy qubits and use one single
quantum output $q$ to produce several identical quantum inputs $q$. This
contradicts the {\em no-cloning theorem} (Theorem~\ref{thm:cloning}).
Therefore, we globally constrain the use of the rule $\textsc{ID}_\textsf{Q}(q,\mathrm{id})$
as shown below.

\begin{abstraction}
  We extend the deductive system with a new set of pairs of terms $\Delta$ in the judgment, thus of the form
  $\Gamma; \Delta\vdash q$.
  We introduce a new deductive system whose judgments are $\Gamma; \Delta\vdashQ q$.
  When a quantum output is triggered, the corresponding term $q$ is added to
  $\Delta$, along with a unique identifier $\mathrm{id}$: $\Delta := \Delta, (q,\mathrm{id})$.
  The attacker can choose a term $q$ as a protocol's input on a quantum channel,
  provided that he can deduce it from its current knowledge: $\Gamma; \Delta\vdashQ q$.
  We also introduce the additional rule $\textsc{ID}_\textsf{Q}$.

  Finally, we consider the following restriction on the use of
  $\textsc{ID}_\textsf{Q}(\cdot, \cdot)$:
  for some term $q$ and identifier $\mathrm{id}$,
  there should be at most one rule $\textsc{ID}_\textsf{Q}(q, \mathrm{id})$ in the derivation.
  This constraint is globally lifted to executions.
  \label{abs:copy}
\end{abstraction}
\begin{example}
  Our refined intruder is now capable of simulating the network by forwarding to the intended
  recipients all classical and quantum terms.
  However, he is still not capable of fully emulating some honest roles as he has no
  way to measure qubits (see Bob's role in BB84 QKD for instance).
\end{example}

\subsubsection{Measurements}
\label{DY:quantum:measure}
To give to the intruder at least the honest roles' capabilities, we need to consider
that he can measure qubits in any base he knows. This would also cover a subset of transformations
he can apply on qubits (\eg modification of bases).
Finally, we also let the attacker forge his own qubits provided he can deduce
the classical bit to encode and the base to use for the encoding.

\begin{abstraction}
We extend the deductive system with the following rules:
$$
\infer[\textsc{Measure}]{\Gamma;\Delta \vdash d}{\Gamma;\Delta \vdashQ \qubit(d,b) & \Gamma;\Delta \vdash b}
\quad \raisebox{6.2pt}{\text{and}}\quad
\infer[\textsc{Forge}]{\Gamma;\Delta \vdashQ \qubit(d,b)}{\Gamma;\Delta \vdash d & \Gamma;\Delta \vdash b}
\raisebox{6.2pt}{.}
$$
Note that we do not need to constrain the use of the \textsc{Measure} rule since the production
of qubits to be measured is already constrained (rule $\textsc{Guess}(q,\mathrm{id})$).
Note that the attacker is already able to deduce terms of the form
$\bit(\mathtt{seed}_\mathrm{attacker},\ldots)$ for some deducible $\mathtt{seed}_\mathrm{attacker}$, therefore different from $\mathtt{seed}$,
as $\bit()$ is a function symbol (see the third rule depicted in Figure~\ref{fig:DY-nat}).
\label{abs:measure}
\end{abstraction}

\begin{example}
  Our refined intruder is now capable of performing full MitM attacks (\eg \Cref{at:MiM,at:done}).
  Indeed, the attacker is now able to emulate honest roles'
  guessing (see Section~\ref{sec:DY:proba:proba}) and quantum capabilities.
  However, he is still unable to perform EPR attacks (\eg \Cref{at:EPR,at:EPR-bit}).
\end{example}

Note that one could have added a rule that allows the attacker to measure a qubit $\qubit(d,b)$
in a different base $b'\neq b$, yielding a random bit (Theorem~\ref{thm:dist}).
But such a rule would be subsumed by the fact that the attacker can always produce new fresh data.

Since the combination of two different bases do not yield orthogonal vectors,
Theorem~\ref{thm:dist} implies that
knowing $\qubit(d,b)$ but not $b$
is not enough to guess $d$ with probability ${>}\frac{1}{2}$.
In our model, we consider such a guess impossible (because the number of qubits depends on security parameters).
In the future, we may refine such rules, for instance when the attacker knows
both $\qubit(d,[+])$ and $\qubit(d,[\times])$ or both $\qubit(0,b)$ and $\qubit(1,b)$.

\subsubsection{EPR-Attacks}
\label{sec:DY:qu:EPR}
\newcommand{\epr}{\qubit_{\mathrm{EPR}}}
Finally, we have seen that EPR pairs can be exploited to learn which is the classical bit that has been
measured by another party, provided that the base used in the measurement is known
(see EPR experiment in Section~\ref{sec:back:mecha:epr}
or the super dense coding and the quantum teleportation scheme in Section~\ref{sec:proto:other}).
We would like to consider an attacker who can exploit this mechanism.
This involves modifications in the protocol semantics since a measurement performed by an honest agent
can potentially increase the attacker's knowledge.

We model an EPR pair of intricated qubits forged by the attacker by a term
$\epr(b,d,\mathrm{id})$ that represents the half of the EPR pair that the attacker may
send to honest agents. The EPR identifier $\mathrm{id}$ aims at keeping track of qubits that may
increase the intruder's knowledge when measured.

We extend the two deductive systems $\vdash$ and $\vdashQ$ as follows.
First, we consider judgments of the form $\Gamma;\Delta;S\vdash M$ and $\Gamma;\Delta;S\vdashQ M$ where $S$ is a partial mapping from
EPR identifiers to $\terms\times\terms$.
A pair of terms $(n,b)$ associated to a an EPR identifier indicates
that the ``honest half'' of the pair has been measured in a base $b$ and yielded $n$.
Second, we consider the additional rule:
$$
\infer[\textsc{Epr}]{\Gamma;\Delta;S \vdashQ \epr(b,d,\mathrm{id})}{\Gamma;\Delta;S \vdash b & \Gamma;\Delta;S \vdash d & \mathrm{fresh\ id}}
\raisebox{6.2pt}{.}
$$
Terms of the form $\epr(b,d,\mathrm{id})$ that the attacker can deduce using the \textsc{Epr} rule
can be chosen by the attacker to be some protocol's input.
When measured by an honest agent, a term $\epr(b,d,\mathrm{id})$ behaves exactly as $\qubit(b,d)$,
except that it additionally produces a substitution $S_\mathit{Sys}=\{\mathrm{id}\mapsto (n,b)\}$
where $n$ is the resulting term from the measurement (\ie $d$ or a fresh name) and $b$ is the base
that has been used to measure the qubit.
The next hypothesis the attacker will have at his disposal when deducing the next inputs
are updated as before, except that $S':=S \circ S_\mathit{Sys}$ (the attacker will know that
his half of the EPR pair $\mathrm{id}$ has been measured and the term $n$ has been observed with regard to the base $b$).
Finally, the deductive systems is enriched with the following rule:
$$
\infer[\textsc{Epr-Leak}]{\Gamma;\Delta;S \vdash n}{\Gamma;\Delta;S \vdash b & \exists\mathrm{id}. (n,b)\in S(\mathrm{id})}
\raisebox{6.2pt}{.}
$$
\begin{example}[Continuing Example~\ref{ex:qubit}]
  Assume that the attacker builds an EPR pair and deduces the term
  $\epr(b_a,d_a,\mathrm{id})$ (where $d_a,b_a$ are terms built by the attacker with $\bit()$ on public constants)
  that he uses as a replacement of the genuine
  qubit $\qubit(\model{d_1},\model{b_1})$ sent by Alice to Bob.
  Upon reception of this input, Bob measures the qubit with regard to $\model{b_1}$ and obtains a fresh name $n$ as a result (because $b_a\not\theo^b \model{b_1}$).
  The new hypothesis $S'$ now maps $\mathrm{id}$ to $(n,\model{b_1})$.
  When Alice sends the bases $(\model{b_i})_i$ to Bob and Bob sends back the matching bases, the attacker is able to deduce $\model{b_1}$
  via a derivation $\Pi$
  (either because $\model{b_1}\in\Gamma'$ or because $\model{\lnot b_1}\in\Gamma'$ so \textsc{Complem} can be used to deduce $\model{b_1}$).
  Therefore, the following derivation shows that the attacker is then able to deduce $n$, the bit Bob has measured:
  $$
  \infer[\textsc{Epr-Leak}]{\Gamma';\Delta';S' \vdash n}{\infer{\Gamma';\Delta';S' \vdash \model{b_1}}{\Pi} & (n,\model{b_1})\in S'(\mathrm{id})}
  \raisebox{6.2pt}{.}
  $$
  By performing this replacement on all qubits sent by Alice to Bob, we can easily capture the EPR attack against
  BB84 QKD (Attack~\ref{at:EPR}) or its bit commitment variant (Attack~\ref{at:EPR-bit}).
\end{example}

\begin{abstraction}
  We introduce a fresh function symbol $\epr(\cdot,\cdot,\cdot)$,
  modify the protocol semantics when operating on such terms,
  and extend the deductive systems with the rules  
  \textsc{Epr} and \textsc{Epr-Leak} as explained above.
  \label{abs:EPR}
\end{abstraction}


\begin{figure}[th] 
  \centering
  \includegraphics[width=1\textwidth,page=2]{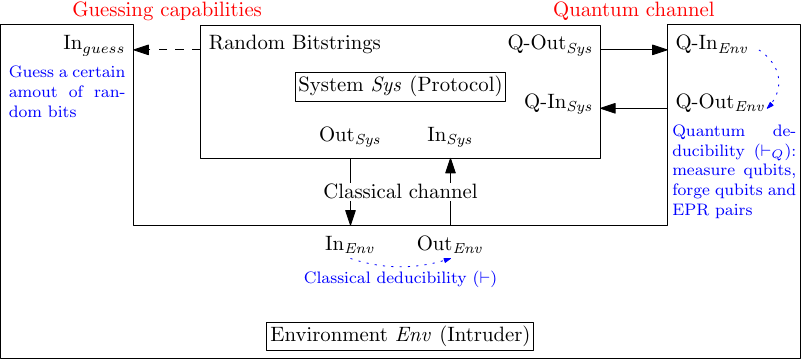}
  \caption{Our Quantum Dolev-Yao Intruder}
  \label{fig:dolevYao:quantum}
\end{figure}

\subsection{Our Quantum Dolev-Yao Intruder}
Figure~\ref{fig:dolevYao:quantum} sums up additional intruder capabilities:
guessing capabilities and a constrained control over the quantum channel.
Our Quantum Dolev-Yao Intruder is simple and abstract enough that he can be embedded
in symbolic verification techniques and tools (see Section~\ref{sec:verif}) but expressive and
powerful enough that he can capture all the aforementioned attacks (\Cref{at:MiM,at:done,at:EPR,at:EPR-bit}).



\section{Formal Verification with Tamarin}
\label{sec:verif}
We first give an overview of the \tamarin verifier (\Cref{sec:verif:tam}) after which we explain how
we can model our abstractions in the tool (\Cref{sec:verif:DY}) and conclude by presenting and discussing
the results of our automated analyses (\Cref{sec:verif:results}).
Finally, we show that similar encoding can be used by competitive verifiers (\proverif, \deepsec, and \akiss)
and we compare their verification efficiency and precision on our case studies(\Cref{sec:verif:other}).

\subsection{The \tamarin verifier}
\label{sec:verif:tam}
This section gives a short and basic introduction to the state-of-the-art \tamarin verifier.
Its content has been adapted from~\cite{5GAKA}.

\subsubsection{Multiset Rewriting Rules}
\label{sec:verif:tam:prem}
\tamarin is a state-of-the-art protocol verification tool for the \emph{symbolic model}, which supports stateful protocols, a high level of automation, and equivalence properties~\cite{tamarin-equiv}.
It has previously been applied to real-world protocols with complex state machines, numerous messages, and complex security properties such as TLS 1.3~\cite{tamarin-tls} or 5G AKA in mobile telephony~\cite{5GAKA}.
We chose \tamarin as it is currently the only tool that combines stateful protocols (mandatory for encoding our abstractions)
and semi-automatic proofs for which proof strategies can be exploited to ease and speed up proof search.
We shall see however that fully automatic tools are also able to verify some of our case studies (see \Cref{sec:verif:other}).

As mentioned earler, in the symbolic model and a fortiori in \tamarin,
messages are described as terms. For example, $\exenc(m,k)$ represents the message $m$ encrypted using the key $k$.
The algebraic properties of the cryptographic functions are then specified using equations over terms.
For example the equation $\exdec(\exenc(m,k),k) = m$ specifies the expected semantics for symmetric encryption:
the decryption using the encryption key yields the plaintext.
As is common in the symbolic model, cryptographic messages do not satisfy other properties than those intended algebraic properties,
yielding the so-called \emph{black box cryptography assumption}
(\eg one cannot exploit potential weaknesses in cryptographic primitives).

The protocol itself is described using multi-set rewrite rules.
These rules manipulate multisets of \emph{facts}, which model the current state of the system with \emph{terms} as arguments.
\begin{example}\label{ex:hashmsr}
The following rules describe a simple protocol that sends an encrypted message.
The first rule creates a new long-term shared key $k$ (the fact $\factStyle{!Ltk}$ is persistent: it can be used as a premise
multiple times).
The second rule describes the agent $A$ who sends a fresh message $m$ together with its MAC with the shared key $k$ to $B$.
Finally, the third rule describes $B$ who is expecting a message and a corresponding MAC with $k$ as input.
Note that the third rule can only be triggered if the input matches the premise, \ie if the input message is correctly MACed with $k$.
\[
\begin{array}{l}
Create\_Ltk : ~ [\factStyle{Fr}(k)] \rwr[] [\factStyle{!Ltk}(k)], \\
Send\_A : ~ [\factStyle{!Ltk}(k), \factStyle{Fr}(m)] \rwr[\factStyle{Sent}(m)] [\factStyle{Out}(\pair{m}{\mac(m,k)})], \\
Receive\_B : ~ [\factStyle{!Ltk}(k), \factStyle{In}(\pair{x}{\mac(x,k)})] \rwr[\factStyle{Received}(x)] [] \qed\\
\end{array}
\]
\end{example}
These rules yield a labeled transition system describing the possible protocol executions
(see~\cite{tamarin-manual,schmidt2012automated} for  the syntax and semantics).
\tamarin combines the protocol semantics with a Dolev-Yao style attacker.
This attacker controls the entire network and can thereby intercept, delete, modify, delay, inject, and build new messages.
However, he is limited by the cryptography: he cannot forge signatures or decrypt messages without knowing the key (black box  cryptography assumption).
He nevertheless can apply any function (e.g., hashing, XOR, encryption, pairing, \ldots) on messages he knows to compute new messages.

\subsubsection{Formalizing Security Goals in Tamarin}
\label{sec:formal:prop}

In \tamarin, security properties are specified in two different ways.
First, trace properties, such as secrecy or variants of authentication, are specified using formulas in a first-order logic with timepoints.
\begin{example}\label{ex:property}
Consider the multiset rewrite rules given in Example~\ref{ex:hashmsr}.
The following property specifies a form of non-injective agreement on the message,
\ie that any message received by $B$ was previously sent by $A$:
\\[0pt]\null\hfill$
    \forall i,m.
    \factStyle{Received}(m)@i
    \Rightarrow ( \exists j .  \factStyle{Sent}(m)@j \wedge j \lessdot i) .
$\hfill\null\\[2pt]
\end{example}
For each specified property, \tamarin checks that the property holds for all possible protocol executions, and all possible adversary behaviors.
To achieve this, \tamarin explores all possible executions in a
backward manner, searching for reachable attack states, which are counterexamples to the security properties.

Equivalence properties, such as unlinkability, are expressed by requiring that two instances of the protocol cannot be distinguished by the attacker.
Such properties are specified using \emph{diff}-terms
(which take two arguments), essentially defining two different instances of the protocol that only differ in some terms.
\tamarin then checks observational equivalence (see~\cite{tamarin-equiv}), i.e., it compares the two resulting systems and checks that the attacker cannot distinguish them for any execution and for any of its behaviors.

\looseness=-1
In fully automatic mode, \tamarin  either returns a proof that the property holds, or a counterexample/attack if the property is violated, or may not terminate as the underlying problem is undecidable.
\tamarin can also be used in interactive mode, where the user can guide the proof search.
Moreover the user can supply heuristics called \emph{oracles} to guide the proof search in a sound way.
We heavily rely on heuristics in our analyses as they allow us to tame the complexity of the protocol, as explained below.

\begin{figure}[th]
\centering 
\begin{lstlisting}[breaklines]
rule Setup:
   let
// Fixed bitstrings (see Example 4)   
       b1   = bit(~k, 'b', '1', 'Alice', '0') // MATCH, rectilinear=0
       b1_  = bit(~k, 'b', '1', 'Bob',   '0')
       b2   = bit(~k, 'b', '2', 'Alice', '1') // NO-MATCH, diagonal=1
       b2_  = bit(~k, 'b', '2', 'Bob',   '0')
       b3   = bit(~k, 'b', '3', 'Alice', '0') // NO-MATCH
       b3_  = bit(~k, 'b', '3', 'Bob',   '1')
       b4   = bit(~k, 'b', '4', 'Alice', '1') // MATCH
       b4_  = bit(~k, 'b', '4', 'Bob',   '1')
       d1   = bit(~k, 'd', '1', 'Alice', '0') // values '0' are not relevant
       d2   = bit(~k, 'd', '2', 'Alice', '1')       
       d3   = bit(~k, 'd', '3', 'Alice', '1')
       d4   = bit(~k, 'd', '4', 'Alice', '0')
       dsL  = <d1,d2,d3,d4>
       bsL  = <b1,b2,b3,b4>
       bsL_ = <b1_,b2_,b3_,b4_>
       qubitsL = <qubit(d1,b1), qubit(d2,b2), qubit(d3,b3), qubit(d4,b4)>
   in
   [ // Secret bitstrings: private seed ~k
     Fr( ~k )
     // Thread id (and shared secret)
   , Fr( ~tid )
   ]
-->
   [ Alice_0(~tid,bsL,dsL,qubitsL)
   , Bob_0(~tid,bsL_)
   // Seed for the private sample names sets
   , !SecretSampling(~k)
   ]
rule Guess: 			// Rule ID_Q, subject to restrictions
let r = <bitstring, role, position> in
  [ !SecretSampling(~k), In(<bitstring, role, position) ]
--[ Guess(r) ]->                // Guess(r) is subject to restrictions
  [ Out(bit(~k, bitstring, position, role, '0'))
  , Out(bit(~k, bitstring, position, role, '1')) ]
rule Complem: 			// Rule Complem
let r = <bitstring, position, role> in
  [ !SecretSampling(~k)
  , In(bit(~k, bitstring, position, role, value))
  ]
--[ Complem(r) ]->
  [ Out(bit(~k, bitstring, position, role, '0'))
  , Out(bit(~k, bitstring, position, role, '1')) ]
\end{lstlisting}
\caption{Modeling of the fixed bitstrings, the sample names sets (see \Cref{abs:fix}, \Cref{abs:name}, and \Cref{ex:name}), as well as \textsc{Guess($\cdot$)} and \textsc{Complem} rules modeling (see \Cref{abs:guess}). Note that the use of \textsc{Guess($\cdot$)} is constrained by restrictions shown in \Cref{fig:code:restr}.}
\label{fig:code:setup}
\end{figure}

\begin{figure}[h]
\centering 
\begin{lstlisting}[breaklines]
// Restrictions 1 and 2 (Example 6):
restriction GuessProba_restriction1:
    "(All #i #j #k role idx1 idx2 idx3. Guess(<'b',role,idx1>)@i
       & Guess(<'b',role,idx2>)@j & Guess(<'b',role,idx3>)@k
           ==> (#i = #j | #j = #k | #i = #k))"
// Restriction 3 (Example 6):
restriction GuessProba_restriction2:
    "(All #i #j #k role idx1 idx2 idx3. Guess(<'d',role,idx1>)@i
       & Guess(<'d',role,idx2>)@j & Guess(<'d',idx3>)@k
           ==> (#i = #j | #j = #k | #i = #k))"
// Restriction 4 (Example 6):
restriction GuessProba_restriction3:
 "(All #i #j role1 role2 f1 f2. Guess(<f1,role1,'1'>)@i & Guess(<f2,role2,'4'>)@j
           ==> #i = #j)"
restriction GuessProba_restriction4:
 "(All #i #j role1 role2 f1 f2. Guess(<f1,role1,'2'>)@i & Guess(<f2,role2,'3'>)@j
           ==> #i = #j)"
\end{lstlisting}
\caption{Tamarin restrictions modeling constraints of \textsc{Guess($\cdot$)} (see \Cref{abs:guess}, and \Cref{ex:guess}).}
\label{fig:code:restr}
\end{figure}

\begin{figure}[h]
\centering 
\begin{lstlisting}[breaklines]
restriction equality:
  "All x y #i.
    (EqB( x, y ) @ #i)
       ==> (Ex k bitstring position role1 role2 value.
                  x = bit(k,bitstring,position,role1,value) &
                  y = bit(k,bitstring,position,role2,value))"
restriction disequality:
  "All x y #i.
    (NeqB( x, y ) @ #i)
       ==> not(Ex k bitstring position role1 role2 value.
                  x = bit(k,bitstring,position,role1,value) & 
                  y = bit(k,bitstring,position,role2,value))"
\end{lstlisting}
\caption{Tamarin restrictions modeling $\theo^b$ (see \Cref{def:name-equ}, and \Cref{ex:name-equ}).
  Note that $x\theo y$ does not necessarily imply $\factStyle{EqB}(x,y)$ here while we require
  $\approx^b_E \supseteq \approx_E$. This is w.l.o.g. though because $\factStyle{EqB}$ is always used when one of the two arguments is of the form $\bit(\ldots)$ which is not subject to any
  relation in $\theo$.}
\label{fig:code:equality}
\end{figure}

\begin{figure}[th]
\centering 
\begin{lstlisting}[breaklines]
functions: qubit/2, [...]   // new function symbol for qubit of arity 2

/*********** Honest roles sending and receiving qubits ****************/
// Rule that lets Alice sends the qubits c1, c2, c3 and c4
rule Alice_send_ds:
  let qubitsL = <c1,c2,c3,c4>
  in
  [ Alice_0(~tid,bsL,dsL,qubitsL) ]
--[ QS('1',c1), QS('2',c2), QS('3',c3), QS('4',c4)
  ]->
// The QSend(number,qubit) *linear* facts store a qubit that has been sent by an honest party
  [ QSend('1',c1), QSend('2',c2), QSend('3',c3), QSend('4',c4)
  , Alice_1(~tid,bsL,dsL)
  ]

// Rule that lets Bob choose bases to measure qubits
rule Bob_receive_connection_ds:
  let bsL_ = <b1_,b2_,b3_,b4_>
  in
  [ Bob_0(~tid,bsL_) ]
  -->
// The QReadBob(number,base) facts store the base that Bob will use to measure the number'th qubit
  [ QReadBob('1',b1_), QReadBob('2',b2_), QReadBob('3',b3_), QReadBob('4',b4_)
  , Bob_0_co(~tid,bsL_)
  ]

/********* Intruder forwarding qubits on the quantum channels *********/
rule q_receive_match_Bob: // Bob measure in the correct base
  [ QSend(id, qubit(d, b1))
  , QReadBob(id, b2)
  ]
--[ QRead(), Forward(), EqB(b1,b2) ]->
  [ QReceiveBob(id,d) ] // QReceiveBob() facts are then given to Bob

rule q_receive_Nomatch_Bob: // Bob measures with a wrong base
  [ Fr(~random)
  , In(id)
  ]
--[ QRead(), NoMatch() ]->
  [ QReceiveBob(id,~random) ] // QReceiveBob() facts are then given to Bob
\end{lstlisting}
\caption{Modeling of qubits, of the quantum channel (\Cref{abs:qubit}) and of the constraint that
  qubits cannot be copied (\Cref{abs:copy}); using linear facts. The fact $\texttt{EqB}(\cdot,\cdot)$ is constrained as shown in \Cref{fig:code:equality}.}
\label{fig:code:qubit}
\end{figure}

\begin{figure}
\centering 
\begin{lstlisting}[breaklines]
rule q_receive_match_Eve: // Intruder measures in the correct base
  [ QSend(id, qubit(d, b1))
  , In( b2 )
  ]
--[ QInterceptMatchEve(id,d,b1), QRead(), EqB(b1,b2) ]->
  [ Out( d ) ]

rule q_receive_forgeEve_match_Bob: // Intruder forges a qubit measured by Bob in the correct base. The rule q_receive_Nomatch_Bob already covers the other case.
  [ In( d )
  , In( b1 )
  , QReadBob(id, b2)
  ]
--[ QRead(), QSendEve(id,qubit(d,b2)), EqB(b1,b2) ]->
[ QReceiveBob(id,d) ]
\end{lstlisting}
\caption{Modeling of qubits measurements (\Cref{abs:measure}) and creation.}
\label{fig:code:measure}
\end{figure}

\begin{figure}[h]
\centering 
\begin{lstlisting}[breaklines]
rule q_receive_Nomatch_Bob_EPR:  // Intruder forges a new EPR pair measured by Bob in a wrong base. The other case is already covered by the rule q_receive_match_Eve.
  [ QReadBob(id, b_)
  , Fr( ~random )
  ]
--[ QRead(), NoMatch() ]->
  [ QReceiveBob(id,~random)
  , State_EPR(~random,b_)
  ]
rule q_measure_EPR_Eve:  // Leak of the measured bit
  [ State_EPR(d, b1)
  , In(b2)
  ]
--[ EPR(), EqB(b1,b2) ]->
  [ Out(d) ]
\end{lstlisting}
\caption{Modeling of Intruder EPR pair creation (\Cref{abs:EPR}).}
\label{fig:code:EPR}
\end{figure}


\subsection{Quantum Dolev-Yao Intruder Modeled in Tamarin}
\label{sec:verif:DY}
We now describe how the non-standard extension of the Dolev-Yao attacker we have described in \Cref{sec:DY} can be
modeled in the \tamarin verifier. We use our model of BB84 QKD to exemplify our modeling choices.

\subsubsection{Modeling Possibilities and Guessing Capabilities}
We depict in \Cref{fig:code:setup} and in \Cref{fig:code:restr} the Tamarin rules that model
the intruder's guessing capabilities and the fixed bitstrings for BB84 QKD (introduced in \Cref{ex:name}).

\subsubsection{Modeling Quantum Intruder}
We depict in \Cref{fig:code:qubit,fig:code:measure,fig:code:EPR,fig:code:equality} how the quantum channel and the intruder's quantum capabilities are modeled in Tamarin, following \Cref{abs:qubit,abs:copy,abs:measure,abs:EPR}.

\paragraph{Qubits and quantum channel.}
\Cref{fig:code:qubit} shows that, when an honest party wants to send a series of qubits $\qubit(d_i,b_i)$,
a series of {\em linear facts} $\textrm{QSend}(i,\qubit(d_i,b_i))$ are created (\Cref{fig:code:qubit},line 12).
Those facts, being linear, cannot be copied. This is how we model the constraint on the use of the rule $\textsc{\textsc{ID}$_\textsf{Q}$}(\cdot)$ (\Cref{abs:copy}).

Next, the recipient (Bob in our example) can commit on bases he wants to use for measuring the incoming qubits, producing
states $\textrm{QReadBob}(i,b_i)$ (\Cref{fig:code:qubit}, line 23).

The intruder can then forward sent qubits (stored in $\textrm{QSend}(i,\qubit(d_i,b_i))$ facts), without interacting with
them, using the \textrm{q\_receive\_match\_Bob} and \textrm{q\_receive\_Nomatch\_Bob} rules. The former
can be triggered when the base used to measure matches with the base in the qubit. The latter can be used otherwise
and produce a random outcome (\Cref{fig:code:qubit}, line 36). In both cases, the outcome is stored in a \textrm{QReceiveBob} linear fact that
can then be read by the appropriate recipient (here Bob).

\paragraph{Intruder's interception capabilities.}
\Cref{fig:code:measure} shows how forging and measuring capabilities are modeled in Tamarin (\Cref{abs:measure}).
As mentioned in \Cref{DY:quantum:measure}, we do not need to give to the intruder an explicit extra capability for
measuring qubits in non-matching bases, as the outcome would be random anyway.

\paragraph{Intruder's EPR capabilities.}
Finally, \Cref{fig:code:EPR} shows our modeling of EPR pair creations and measurements. The intruder can forge new EPR pairs and send them to Bob with the rule \textrm{q\_receive\_noMatch\_Bob\_EPR}. Note that this rule only covers the case where Bob uses a wrong base to measure the qubit. Indeed, the case where Bob uses the appropriate base is already covered by the rule \textrm{q\_receive\_match\_Bob}.
When the associated qubit $\qubit(d,b)$ is received and measured by an Honest party (here Bob) in a base $b'$,
an additional persistent state $\textrm{State\_EPR}(n,b')$ is produced, where $n$ is what Bob has measured.
This fact has the role of $S_\mathit{Sys}$ (\Cref{sec:DY:qu:EPR}) as it can be read by the intruder provided he can deduce $b'$ (\Cref{fig:code:EPR}, rule \textrm{q\_measure\_EPR\_Eve} at line 14).

\subsection{Analyzed Scenarios and Results}
\label{sec:verif:results}
We have modeled the BB84 QKD and QBC protocols in the \tamarin verifier.
Our \tamarin models are freely available at~\cite{models}. We were able to automatically obtain security proofs
(w.r.t.~our model) and automatically finding attacks from those models.

However, \tamarin needs a considerable amount of time for verifying BB84 QKD involving four qubits when considering the full quantum Dolev Yao attacker.
While we were able to verify the full protocol with all attacker's capabilities in less than
one hour, we adopt a more efficient, yet generic, methodology by verifying the protocol
for increasingly complex settings (\eg more and more qubits)
and for increasingly rich threat models; \ie considering more and more of the attacker capabilities we have defined (\eg forging, guessing, EPR).
We thus first discover attacks very quickly for simpler models that also affect the most complex model we have.
We then show when security holds, for example under stronger trust assumptions, and then
gradually add attacker capabilities or increase the number of qubits.
This allows for a quick verification-fix loop until reaching the target protocol and threat model.
We stress that the different threat models we consider are defined in a fully modular way
and can be reused to analyze other protocols.

\subsubsection{Verification of BB84 QKD}
\label{sec:verif:proto:BB84}
We have already described most parts of our modeling of BB84 QKD for the scenario described in \Cref{ex:qubit}
through \Cref{fig:code:setup,fig:code:restr,fig:code:qubit,fig:code:measure,fig:code:EPR,fig:code:equality}.

\paragraph{Security Goals.}
\looseness=-1
We are interested in verifying secrecy on the session key established by the protocol from Alice's point of view
and from Bob's point of view.
We leave the verification of agreement on the key as future work, but there is no conceptual
issue that would prevent us to do so.

\paragraph{Increasingly powerful attackers.}
Formally we consider the following threat models:
\begin{itemize}
\item {\em Passive attacker}: the attacker cannot guess any bit and cannot forge or measure any qubit. He can only forward qubits.
  This corresponds to a quantum attacker with only the rule \textsc{ID$_\textsf{Q}(\cdot)$} but no rule
  \textsc{Complem}, \textsc{Guess$(\cdot)$}, \textsc{Measure}, \textsc{Forge}, \textsc{Epr} or \textsc{Epr-Leak} (see \Cref{sec:DY}).
\item {\em Forge attacker}: as the passive attacker but with the additional capabilities of
  measuring qubits (rule \textsc{Measure}) and forging qubits (rule \textsc{Forge}).
\item {\em EPR attacker}: as the forge attacker but with EPR forging capabilities (rules \textsc{Epr} and \textsc{Epr-Leak}).
  Note that we slightly and soundly modified the way EPR capabilities are modeled in Tamarin (from \Cref{fig:code:EPR}).
  Indeed, we consider the worst case scenario where the attacker always produce EPR pairs when forging qubits.
  This is w.l.o.g.
\item {\em Guess attacker}: as the forge attacker but with guessing capabilities (rule \textsc{Guess$(\cdot)$} and \textsc{Complem}).
\item {\em Full attacker}: as the guess attacker but with EPR forging capabilities (rule \textsc{Epr} and \textsc{Epr-Leak})
  {\em and} guessing capabilities (rule \textsc{Guess$(\cdot)$} and \textsc{Complem}).
\end{itemize}
Note that the full attacker is the quantum Dolev Yao attacker we have defined in \Cref{sec:DY}. However,
the passive attacker does not correspond to the classical Dolev Yao attacker as probabilities are already handled
differently (see \Cref{abs:name} for instance).

\newcommand{\TMorder}{\textsf{Order}\xspace}
\newcommand{\TMdone}{\textsf{Auth(done)}\xspace}
\newcommand{\TMbasesAlice}{\textsf{Auth(bases)}\xspace}
\newcommand{\TMbases}{\textsf{Auth(matchingBases)}\xspace}
\newcommand{\TMdonebases}{\textsf{Auth(done,matchingBases)}\xspace}
\newcommand{\TMdonebasesverif}{\textsf{Auth(done,matchingBases,verif)}\xspace}
\newcommand{\TMverif}{\textsf{Auth(verif)}\xspace}
We have modeled the five different threat models for different compromised scenarios
in a modular way in a single file (\texttt{QKD\_BB84.m4} from~\cite{models}) by using the \texttt{m4} macro processor.
From this single file, the five different \tamarin models can be automatically generated.
In addition to the aforementioned variants of threat models,
we also model compromised scenarios corresponding to different trust assumptions on
the classical channels: \eg are the matching bases sent by Bob to Alice authenticated,
are all the verification bits send by Alice to Bob authenticated, \etc 
This is also implemented in a modular way by labeling the rules that allow such compromises.
The resulting models, proofs, attacks, and instructions for reproducibility can be found at~\cite{models}.

\Cref{tab:sources} notably depicts the number of sources\footnote{This number gives an idea about the size of the search space as it corresponds to the number of cases to be considered for the rules be triggered.} as computed by Tamarin for the different threat models.
For instance, there are 989 different sources for the Full attacker which explains why verification based on such a model
can take a dozen of minutes to perform.
\begin{table}  [th]
  \centering
  \begin{tabular}{r|c|c|c|c}
    Threat Model&\# Rules & \# Sources & Max \# Sources (per rule) & \# Sources for rule Bob\_1 \\\hline
    Passive & 21 & 438 & 29 & 20 \\
    Forge  & 21 & 2006 & 492 & 81 \\
    EPR & 22 & 2017 & 492 & 81 \\
    Guess  & 21 & 2010 & 492 & 81 \\
    Full & 22 & 2021 & 492 & 81 \\
  \end{tabular}
  \caption{Number of sources of the different threat models. We indicate the number of rules, the number of total sources, the maximum number sources of a single rule, as well as the number of sources of the rule Bob\_2 which corresponds to the final Bob input.}
  \label{tab:sources}
\end{table}

\paragraph{Results.}
We now summarize our result, starting with the weakest and ending with the strongest threat models.
\smallskip{}

\noindent{\em Passive attacker.}
As expected, we were able to show that key secrecy from both point of views hold without any authenticity assumption.
\smallskip{}

\noindent{\em Forge attacker.}
We state and analyze 3 lemmas for checking under which conditions, such
as authenticity of messages on classical channels, key secrecy is met.
Incidentally, we automatically have found \Cref{at:done} with Tamarin.
For key secrecy to hold from Bob's point of view, it is required that either
the authenticity of the message \texttt{Done} is provided (denoted by \TMdone) or
all Bob's qubits measurements should happen before the Alice's bases reveal
(\ie {\em strict ordering} of measurements and bases reveal, denoted by \TMorder).
In particular, we show that for that threat model, when
\TMorder is enforced but no message sent on the classical channel is authenticated,
secrecy of the established key holds from Bob's point of views.
Similarly, we formally show that key secrecy from Alice's point of view requires \TMorder.
However, secrecy from Alice's point of view fails to hold as soon as \TMorder
is violated, even when \TMdone is provided.
This is as expected since Alice may then think she has established a secure key
with Bob while Bob has not yet confirmed he has received matching verification
bits, which will fail.
\smallskip{}

\noindent{\em Guess attacker.}
Obviously, we start by adding a pre-condition to the secrecy lemmas by forbidding the attacker to guess
the bit $d_i$ for which we check secrecy.
We show that enforcing \TMdone is not sufficient any more
to achieve secrecy from Bob's point of view. Indeed, the attacker can leverage his guessing
capabilities to create discrepancies between the sets of matching bases from Alice's and Bob's point of view.
We show however that when \TMorder is enforced, then secrecy still holds from Alice's and Bob's point of view.
Similarly, we show that when \texttt{Done} and the matching bases are authentic (denoted by \TMdonebases),
secrecy holds, but only from Bob's point of view.
\smallskip{}

\noindent{\em EPR attacker.}
We show that even when \TMorder or \TMdone is enforced,
key secrecy from Bob's point of view is no longer satisfied, as opposed to the {\em Forge attacker}.
Indeed, we automatically found \Cref{at:EPR} witnessing the latter.
We then show that even when \texttt{Done} and the matching bases are authentic (denoted by \TMdonebases),
secrecy fails to hold  (we found \Cref{at:EPR} otherwise).
It is required that, in addition,
the authenticity of the verification bits sent by Alice on the classical channel is ensured.
Incidentally, in our model, this also implies that all bits corresponding to matching bases, and not only the verification bits, should be
checked against authenticated, received bits. Otherwise, the attacker could exploit his knowledge of the bit chosen
for verification to perform targeted EPR-attack that would happen with negligible probability without our abstractions.
We modeled this extra check through an additional exchanged bit sent over a compromised secure channel (attacker has write but not read access)
from Alice to Bob. We write \TMverif when authenticity is provided to the verification bits and to the bits chosen for the key.
We were then able to prove secrecy under this threat model; \ie \TMdonebasesverif.
Hence \TMverif is a necessary condition.
\smallskip{}

\noindent{\em Full attacker.}
We have not found additional attacks that require all the previous attacker capabilities at the same time.
Indeed, we were able to prove secrecy under the combination of the necessary conditions mentioned so far.
To sum up, we were able to show secrecy from Bob's point of view under those minimal security assumptions:
\begin{enumerate}
\item All bits $(d_i)$ should be checked for equality;
  \ie the bits of the key should be implicitly checked in order to avoid bitstring-specific attacks.
  While this could be interpreted as a modeling artifact, it also shows an expected weakness when not enough verification
  bits are checked, allowing key bits-targeted attacks.
\item \TMdonebasesverif should be enforced.
\end{enumerate}
In particular, the bases sent by Alice to Bob do not have to be authentic in our model, \ie
\TMbasesAlice is not required.

Secrecy from Alice's point of view requires \TMorder, which in our opinion should rather be realized through
cryptography. We believe that by analyzing the protocol with the key-confirmation phase,
Alice could then obtain secrecy under weaker assumptions. We leave this task as future work.



The above conclusion follows from the verification of 26 lemmas that we have conduced fully automatically using \tamarin.
The total computation time is about 2 hours with 16 cores Intel Xeon 3.10GHz and \tamarin, branch \texttt{develop}, version \texttt{1.5.1}.

\paragraph{Results for a simpler scenario with 2 qubits.}
We first have verified a simplified scenario with 2 qubits for which the bases $b_i$ and $b_i'$ match (between Alice and Bob).
The first one will be used for the secret key
and the second one for the verification bit.
Unsurprisingly, this simplification considerably reduced the verification time.
This is supported by the comparison of the number of sources shown in \Cref{tab:sources-scenario}.
The resulting model is provided in the file \texttt{QKD\_BB84\_2qubits.spthy} from~\cite{models}.
We were already able to capture most of the aforementioned attacks in this simpler model.


\begin{table}[th]
  \centering
  \begin{tabular}{r|c|c|c|c}
    Scenario l&\# Rules & \# Sources & Max \# Sources (per rule) & \# Sources for rule Bob\_1 \\\hline
    4 qubits & 22 & 2021 & 492 & 81 \\
    2 qubits & 22 & 116 & 21 & 9 \\
  \end{tabular}
  \caption{Number of sources of the different scenarios considered for the full  quantum attacker.}
  \label{tab:sources-scenario}
\end{table}


\paragraph{Efficiency issues.}
We now give some explanations about the verification time (2 hours) that could be considered quite long, considering all our abstractions.
The main reason is that all our attacker's capabilities can be combined in many ways on a single qubit, and it gets much worse when considering
multiple qubits.
For instance, any single qubit that is sent can be either forwarded, measured in an appropriate base and an equal forged
qubit is resent, measured in a wrong base and a new qubit is forged instead, measured in an appropriate base and 
an EPR pair of qubits is created and half of it is sent instead, \etc\ 
When this is explored for all four qubits, Tamarin experiences the expected combinatorial explosion
(see \Cref{tab:sources-scenario}).
In practice, we reach the Tamarin verifier limits in terms of efficiency when considering more qubits (\eg we tested 6).
We stress however that 4 qubits is already enough to capture many interesting scenarios and to quickly obtain a
certain level of security guarantees.

We conjecture that there surely are a lot of redundancies in those explorations and one could come up with
sound restrictions to mitigate this explosion.
We leave this as future work.

\subsubsection{Verification of BB84 QBC}
\label{sec:verif:proto:bitcomm}
We focus on the binding property of the bit commitment scheme.
Namely, we are interested in proving that a malicious Alice cannot pretend to have committed on a different base:
if Bob accepts the Alice's final claim then \texttt{base} in the unveil procedure must be the same as \texttt{base}
in the commit procedure. We thus consider that Alice is entirely controlled by the intruder.

\paragraph{Modeling the binding property.}
Since the intruder can produce any data to be sent to Bob during the commit procedure,
capturing the intruder's intent to commit on a value during that procedure is non-trivial.
For instance, using the messages sent by the intruder to Bob to do so constrains the attacker
to send messages of specific formats (\eg a series of qubits encoded in an uniform base linked to the bit).

We took a different approach: we reveal a fresh value {\em after the commit procedure} (that was previously secret
for the attacker) and check if the intruder is able to make Bob believes he has committed on that value.

\paragraph{Modeling Choices.}
We fix the security parameter to 4 which is enough to explore interesting executions and find \Cref{at:EPR-bit}.
The four bases randomly chosen by Bob are taken from the same sets, and we make them interchangeable. 
Since the possible values on which the intruder can commit on are in $\mathbb{B}$, we fix the scenario
where the intruder does not known the base corresponding to $[+]$ (but knows $[x]$) until the end of the commit procedure
while it must convince Bob he has committed on $[+]$ at the end of the unveil procedure.
We obviously adapt our guessing capabilities accordingly by removing the \textsc{Complem} rule and by restricting our
\textsc{Guessing} rule.

As part from that, we model the quantum channels and intruder's capabilities as explained in \Cref{sec:DY:qu}.
Those parts of the model are exactly the same between the BB84 QKD (\Cref{sec:verif:proto:BB84}) and
QBC (\Cref{sec:verif:proto:bitcomm}) models, witnessing the genericity of our modeling choices.

The resulting model can be found in the file \texttt{QBC\_BB84.spthy} from~\cite{models}.

\paragraph{Results.}
We check the binding property for two intruder models:
\begin{enumerate}
\item {\em Guess attacker}: the intruder model described in \Cref{sec:DY} excluding EPR capabilities (\ie excluding \Cref{sec:DY:qu:EPR}).
\item {\em Full attacker}: the intruder model described in \Cref{sec:DY}.
\end{enumerate}
Using \tamarin, we were able to automatically find {\em the} EPR attack described in \Cref{at:EPR} for the full attacker (2) and
automatically prove that the scheme is secure against an attacker who cannot forge EPR pairs (1) for the scenario we consider here.
Both results are automatically obtained in a couple of seconds. We also have analyzed those properties for a
smaller scenario with 2 qubits (see \texttt{QBC\_BB84.spthy} from~\cite{models})
and obtained similar results instantly.
The \tamarin representation of the EPR attack for two qubits against the threat model (2) found automatically is depicted in \Cref{fig:at:EPR}.

\begin{figure}[ht]
  \centering
  \includegraphics[width=1\textwidth,page=1]{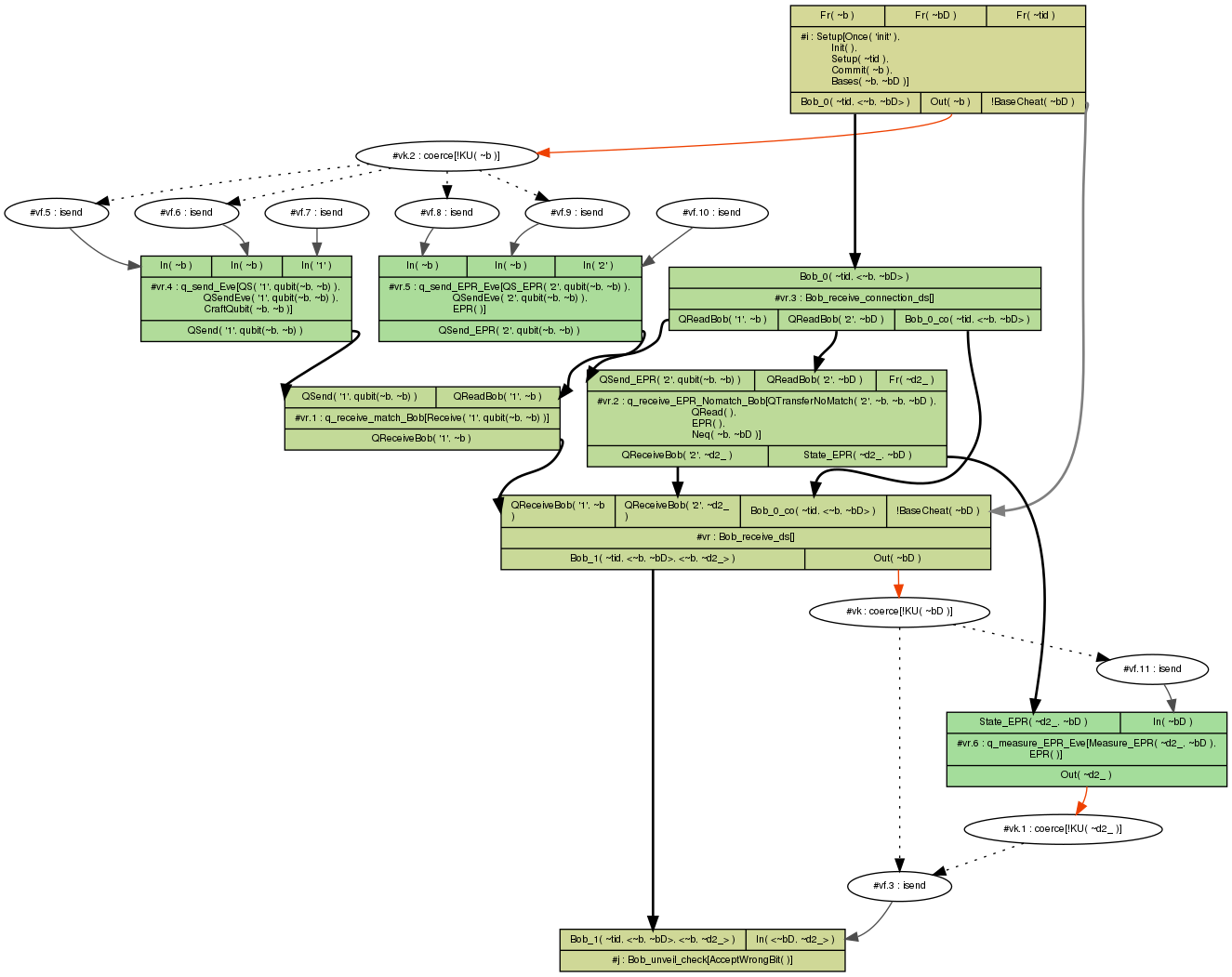}
  \caption{EPR attack as found by \tamarin.}
  \label{fig:at:EPR}
\end{figure}

\subsection{Verification With Other Tools}
\label{sec:verif:other}
We also have explored the use of other state-of-the-art tools to analyze our case studies.
We notably have modeled the BB84 QBC protocol with 4 qubits in \proverif~\cite{PROVERIF}, \deepsec~\cite{DEEPSEC}, and \akiss~\cite{AKISS}. All our models are available at~\cite{models}.
Our goal here is only to show that there is no fundamental problem in modeling and verifying quantum protocols
using our abstractions in those other tools.

\proverif was not able to terminate\footnote{A new version of \proverif currently under development that notably aims at reducing the number of false attacks was able to terminate in seconds. We do not discuss those results as they cannot be reproduced at the moment.} when using the last public version \texttt{2.00} or when using the extension \textsc{GsVerif}~\cite{cheval2018little}.

Note that \akiss and \deepsec are tools that decide a notion of behavioral equivalence for a bounded number of sessions.
The restriction to a bounded number of sessions is not a limitation in our case as we only deal with a bounded number of sessions
of Bob and Alice, thus dealing with a bounded number of classical bits and of exchanged qubits that only require a bounded number of attacker processes handling quantum and probability capabilities.
While those tools do not check for reachability properties, we were able to encoded the aforementioned security properties,
taking the form of reachability predicates, into behavioral equivalence properties; using encoding in the folklore.

\akiss (last development version as time of writing) terminates in about 3 minutes with 1 core (\akiss is single-threaded)
while \deepsec \texttt{version 1.0.0} terminates in about 5 minutes with 16 cores. We expect the verification time
of larger examples (\eg BB84 QKD with 4 qubits) to blow up but we leave this investigation as future work.



\section{Conclusion}
\label{sec:conclu}
We have explored how symbolic models and the Dolev Yao attacker can be extended for
modeling quantum protocols. We have proposed such an extension, balancing the trade-off between
precision with regards to quantum physics and the level of automation we can leverage using existing
symbolic verifiers.
We have explained how our extended model can be encoded in the \tamarin verifier and
have evaluated our trade-offs on the well-known BB84 QKD and QBC quantum protocols.
The results we have obtained show that our model is precise enough to capture known attacks
and explore different threat models in order to identify minimal security assumptions.

This leaves several exciting avenues for future research.
First, the lack of scalability (\eg dealing with a dozen of qubits) is certainly a limitation.
This raises the question of how this can be avoided by developing
finer and less costly encodings or mitigated by reducing the large amount of explorations
that are potentially redundant (\eg symmetry reductions, partial order reduction).
Next, we would like to analyze more recent quantum protocols
(such as~\cite{huang2014cryptanalysis,shukla2017semi}) in order to evaluate
the capacity of our method to find (ideally new) attacks, even for simple scenarios.
Finally, we view this work as a first tentative step to define a quantum symbolic model.
There are certainly different and possibly better ways to balance precision and efficiency.
For instance, it would certainly be interesting to explore extensions that model
lower level quantum properties and mechanisms (\eg atomic transformations on qubits)
rather than built-in high level principles (\eg EPR pairs creation and measurement)
as done in the present work.


\subsection*{Acknwoledgements}
The author would like to thank David Basin and Ralf Sasse for their
helpful comments and suggestions on earlier drafts of this paper
and to Renato Renner for his quantum expert feedback.

\bibliographystyle{abbrv}
\bibliography{bib-q,5GAKA,bib,bib-intro,bib-survey} 

\appendix
\section{Bell states}
\label{ap:Bell}
Consider the following circuit acting on $\mathbb{C}^2$
(the first gate is the Hadamard gate and the second gate is the C-NOT gate):
$$\Qcircuit @C=1em @R=.7em {
  & \gate{H} & \ctrl{1} & \qw \\
  & \qw & \targ &  \qw
  }$$

The circuit above maps:
\begin{itemize}
\item $\ket{11}$ to $\ket{\Phi^-}\deff= \frac{\ket{01} - \ket{10}}{\sqrt{2}}$ (the Bell state),
\item $\ket{01}$ to $\ket{\Phi^+}\deff= \frac{\ket{01} + \ket{10}}{\sqrt{2}}$,
\item $\ket{10}$ to $\ket{\Psi^-}\deff= \frac{\ket{00} - \ket{11}}{\sqrt{2}}$,
\item $\ket{00}$ to $\ket{\Psi^+}\deff=\frac{\ket{00} + \ket{11}}{\sqrt{2}}$.
\end{itemize}
Note that $\{\ket{\Phi^-}, \ket{\Phi^+}, \ket{\Psi^-}, \ket{\Psi^+}\}$ forms an orthonormal basis of $\mathbb{C}^2$, called
the {\em Bell basis}.


\end{document}